\documentclass[journal]{IEEEtran}
\ifCLASSINFOpdf
\else
\fi
\usepackage{cite}
\usepackage{amsmath,graphicx,amssymb,color,textcomp,amsfonts,subfigure,enumerate,bm}
\usepackage{extarrows}
\usepackage{makecell,multirow,diagbox,stfloats}
\newcommand{\reffig}[1]{Fig. \ref{#1}}

\renewcommand\arraystretch{1.5}
\usepackage{amsthm}
\usepackage{cases}
\usepackage{algorithmicx,algorithm,balance}
\usepackage{algpseudocode}
\usepackage{colortbl}
\usepackage{array}
\usepackage{setspace}
\usepackage{booktabs}
\usepackage{cuted}
\usepackage{forloop}
\usepackage{bbm}
\newcounter{ct}

\newtheorem{definition}{Definition}
\newtheorem{proposition}{Proposition}

\begin{document}
	\title{{%Fast Connectivity Self-healing Mechanism of Resilient USNET System to Unpredictable Disruptions %based on Hybrid of Graph Convolution and Mean Field Reinforcement Learning
	Millimeter Wave Wireless Communication Assisted Three-Dimensional  Simultaneous Localization and Mapping
	}}
	
		\author{
		Zhiyu Mou, and Feifei Gao
		
	%	\thanks{Manuscript received June 15, 2021; revised September 10, 2021; accepted October
		%	18, 2021. 
			%The work of Zhiyu Mou and Feifei Gao was supported in part by National Key Research and Development Program of China (2018AAA0102401), by the National Natural Science Foundation of China under Grant \{61831013,61771274\}, 
		%	This work was supported in part by National Key Research and Development Program of China (2018AAA0102401), by the National Natural Science Foundation of China under Grant 61831013, by Beijing Municipal Natural Science Foundation
		%	under Grant \{L182042, 4212002\}, and by Tsinghua University-China Mobile Communications Group Co.,Ltd. Joint Institute.
		%	This work was also supported in part by the National Natural Science Foundation of China under Grant 61902214.
		%	The work of Qihui Wu was supported by Project of Major Scientific Instrument, Natural Science Foundation of China (NSFC) under Grant 61827801.
		%	\emph{(Corresponding author: Jun Liu, Feifei Gao.)}}
		
		\thanks{
			{Z. Mou, F. Gao are with the Institute for Artificial Intelligence, Tsinghua University (THUAI), State Key Lab of Intelligence Technologies and Systems, Tsinghua University, Beijing National Research Center for Information Science and Technology (BNRist), Department of Automation, School of Information Science and Technology, Tsinghua University, Beijing 100084, China (email: mouzy20@mails.tsinghua.edu.cn, feifeigao@ieee.org). }}
	%	\thanks{{J. Liu is with the Institute of Network Sciences and Cyberspace, Tsinghua University, Beijing 100084, China, and also with the Beijing National Research Center for Information Science and Technology (BNRist), Tsinghua University, Beijing 100084, China (email: juneliu@tsinghua.edu.cn).}}
	%	\thanks{{Q. Wu is with the College of Electronic and Information Engineering, Nanjing University of Aeronautics and Astronautics, Nanjing 210016, China (email: wuqihui2014@sina.com).}}
		%	\thanks{
			%		{Z. Han is with the Department of Electrical and Computer Engineering,
				%	University of Houston, Houston, TX 77004, USA, and also with the
				%		Department of Computer Science and Engineering, Kyung Hee University,
				%	Seoul 02447, South Korea (e-mail: zhan2@uh.edu).}}
	}
	\maketitle	
	\date{}
\begin{abstract}  
	In this paper, we study the three-dimensional (3D) simultaneous localization and mapping (SLAM) problem in complex outdoor and indoor environments based only on millimeter-wave (mmWave) wireless communication signals. 
	Firstly, we propose a deep-learning based mapping (DLM) algorithm 
	that can leverage the reflections point on the first-order none line-of-sight (NLOS) communications links (CLs) to build the 3D point cloud map of the environment. 
	Specifically, we design a classification neural network to identify the first-order NLOS CL and theoretically calculate the geometric coordinates of the reflection points on it. 
	Secondly, we take the advantage of both the inertial measurement unit and the beam-squint assisted localization method to realize real-time and precise localizations.
	Then, combining the DLM and the adopted localization algorithm, we develop the communication-based SLAM (C-SLAM) framework that can carry out SLAM without any prior knowledge on the environment. Moreover, extensive simulations on both complex outdoor and indoor environments validate the effectiveness of our approach.
	\end{abstract}
%\vspace{-0.3cm}

\section{Introduction}
A widely recognized vision for the next generation wireless  communication system, such as beyond 5G (B5G) or 6G networks, is to be combined with sensing systems, realizing efficient utilization wireless resources, wide environment sensing functions, and even pursuing mutual benefits \cite{isac_1, isac_2}. Therefore, integrated sensing and communication (ISAC) is considered as one of the most important and promising technologies in the future communication systems and has attracted increasingly research interest recently \cite{slam, isac_3, ny, bs_l}. Many ISAC technologies have investigated how to improve the efficiency and quality of data transmissions in wireless communication using sensing functions \cite{s_t_c_1, s_t_c_2}. 
For example, \cite{s_t_c_1} tries to select proactive base state and carry out optimal beam switching with the aid of multi-cameras (that are typical sensing equipments). In addition, \cite{s_t_c_2} infers the optimal beam pair for transceivers without any pilot signal overhead based on three-dimensional (3D) object detection techniques. The success of these works has proved to us that sensing can indeed help to improve the communication systems. 
Conversely,
communication signals also have the potential to enable sensing abilities \cite{bs_l, slam}, which has large application prospects in
user searching, emergency rescues, and many other scenarios.
In particular, in the harsh environment with severe dust, thick smoke or darkness that cause low visibilities, traditional optical sensing equipments are seriously affected or even fail. 
Compared with expensive radar systems, low-cost and easily available wireless communication equipments can be ideal tools for environmental sensing \cite{slam}. 

\begin{table}
	\centering
	\caption{Existing Localization methods as well as their advantages and disadvantages. }
	\label{table:localization}
	\renewcommand\arraystretch{1.7}
	\begin{tabular}{c|c|c}
		%	\specialrule{0em}{1pt}{1pt}
		\hline
		\rowcolor[gray]{0.9}
		\small\textbf{{Positioning Methods}}&\small\textbf{{Advantages}}&\small\textbf{{Disadvantages}}\\
		\hline
		\small \makecell[c]{{Global Positioning}\\{ System (GPS)}}&\small
		\makecell[c]{{ability to correct} \\ { errors of other } \\{positioning systems}
		}& \small\makecell[c]{{
				(1) low update } \\{frequency} \\{(2)  low accuracy}} \\
		\hline
		\small \makecell[c]{{
				Inertial Measurement}\\ { Unit (IMU)}} & \small \makecell[c]{{
				high update} \\ { frequency}}& \small \makecell[c]{{
				accumulative} \\ { error }}\\
		\hline
		\small \makecell[c]{{LIDAR with }\\ {HD map} } & \small \makecell[c]{{
				(1) high precision} \\ {(2) high real-time}}& \small \makecell[c]{{ relies on the } \\ {the HD map }}  \\
		\hline
		\small \makecell[c]{{camera with} \\ {the HD map}} & \small high real-time & \small \makecell[c]{{easily affected} \\ { by  lights, etc. }}\\
		\hline
		
	\end{tabular}
	% \vspace{-0.8cm}
\end{table}

Simultaneous localization and mapping (SLAM) is an important and widely studied sensing technology in the field of robotics and autonomous driving \cite{v-slam, l-slam}.  Traditional SLAM algorithms usually leverage LIDAR \cite{l-slam} or computer vision methods \cite{v-slam}. Nonetheless, as a typical ISAC technique that enables sensing functions only with wireless signals, 
\emph{communication-based SLAM} (C-SLAM) has been studied in some recent literatures \cite{slam, slam_bp, slam_bp_1, slam_bp_2, ny, slam_3, slam_4}.
Specifically, in terms of mapping, 
some researchers propose belief propagation (BP)-based SLAM algorithms to detect the physical anchors (PAs) and virtual anchors (VAs) that can represent specific boundaries in simple indoor environments \cite{slam, slam_bp, slam_bp_1, slam_bp_2}. However, the anchors determined by the BP-based SLAM algorithms can only reconstruct simple environments. For complex ones, BP-based SLAM algorithms can only roughly find the boundaries and miss the details of the environment. The authors in \cite{slam_4} leverage mmWave imaging to construct a high definition (HD) 3D image of the line-of-sight (LOS) and non-line-of-sight (NLOS) objects in the environment with one antenna array. However, similar to the BP-based algorithm, \cite{slam_4} only obtains the image of simple objects, i.e., two walls. More importantly, the algorithms in  \cite{slam, slam_bp, slam_bp_1, slam_bp_2, slam_4} can only applied to two-dimensional (2D) environment and fail to capture the details of 3D real-world objects. 
Other researchers focus on designing the hybrid mapping system that requires the use of cameras, which may fail in harsh environments \cite{ny}. Hence, new algorithms should be designed to solve the mapping problem of complex 3D environments with only communication signals. 
As for localization, many mature techniques have been used in realities, whose advantages and disadvantages are summarized in Table. \ref{table:localization}. Specifically, the global positioning system (GPS) \cite{gps} has the ability to correct estimation errors of other localization systems, but usually has low update frequency and accuracy.
Although the inertial measurement unit (IMU) \cite{imu} has high update frequency, it suffers from the accumulative estimation error. Note that the combination of the GPS and the IMU is also not feasible, since the  GPS cannot meet the accuracy requirement for the localizations in C-SLAM algorithms. Other techniques, such as LIDAR and computer vision-based localization algorithms, all need the HD map of the environment in advance. However, the HD map or any other information on the environment may not be available. Hence, new localization method needs to be devised in the C-SLAM algorithms. 
\begin{figure}
	\centering
	\includegraphics[width=80mm]{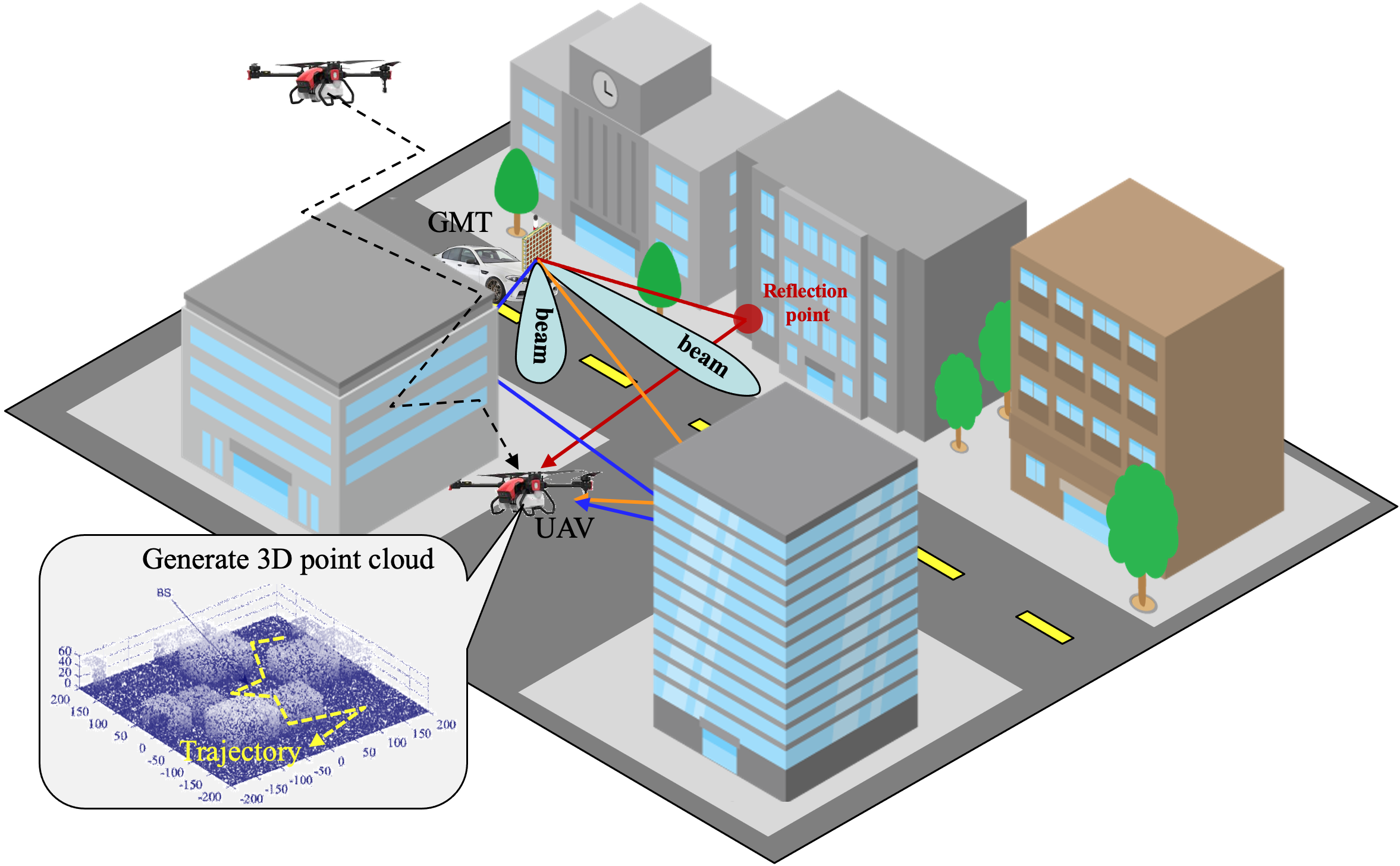}
	\caption{The UAV aims to build the 3D point cloud of the surrounding urban outdoor environment and locates itself only based on the wireless signals received from the GMT.}
	\label{fig:system_model}
\end{figure}

%However, due to the characteristics of complex multipath transmission, energy attenuation, false alarms and missed detection of communication signals, the realization of high precision and  density map reconstruction faces great challenges. In addition, solving the map point cloud depends on channel state information (CSI) such as angle of arrival (AOA), angle of departure (AOD) and time of arrival (TOA) with high resolution. Embedding the traditional point cloud coordinate solution method based on CSI in the communication stage will lead to exponential computational amount and complexity. 

In this paper, we study the 3D SLAM problem in complex outdoor environments based only on millimeter-wave (mmWave) wireless communication signals. 
Firstly, we propose a deep-learning based mapping (DLM) algorithm to build the 3D point cloud map of the environment. The basic idea of the DLM is to solve the coordinates of the reflections point on the first-order none line-of-sight (NLOS) communications links (CLs) and use them as the points in the 3D point cloud map of the environment. To achieve this, the DLM adopts the low-rank tensor decomposition (LTD) method to estimate the channel state information (CSI) of the CLs and designs a link state classification neural network (LSCN) to identify the first-order NLOS CLs. 
Based on the estimation and classification results, the DLM theoretically calculates the geometric coordinates of the reflection points. 
Secondly, 
we take the advantage of both the IMU and the beam-squint assisted localization (BSM) method \cite{bs_l}
to realize real-time and precise localizations. The combined
Moreover, combining the DLM and the adopted localization algorithm, we develop the communication-based SLAM (C-SLAM) framework that can carry out SLAM without any prior knowledge on the environment. Moreover, extensive simulations on both complex outdoor and indoor environments validate the effectiveness of our approach.

The rest parts of this paper are organized as follows.
Section \ref{sec:system_model} presents the system models of the C-SLAM problem. Section \ref{sec:algorithm} describes the proposed C-SLAM framework. Simulation results are presented in Section \ref{sec:simulations},  and conclusions are made
in Section \ref{sec:conclusion}.

\emph{Notations:}  $x$, $\mathbf{x}$, $\mathbf{X}$ represent a scalar $x$, a vector $\mathbf x$ and a matrix $\mathbf X$, respectively; $\sum$ denotes the sum operator; $\left\|\cdot\right\|_2$ denote the norm-2 operation of matrices; $\cup$ represents the union operator  between sets; $\leftarrow$ represents the right-to-left assignment notation;
$^*$ denotes the complex conjugate operator;
$\neq$ means that the left and right values are not equal.
$\|$ is the concatenation operation; 
$\{\cdot\}$ denotes a set, and $|\{\cdot\}|$ calculates the size of the set;
$\mathbb{R}^+$ represents the set of positive real numbers, while $\mathbb{N}_+$ denotes the set of positive integers;
$\in$ represents the element on the left belongs to the set on the right, respectively;  in addition, $\mathbb{R}^n$ represents the set of vectors composed of $n$ real number;
moreover, $\equiv$ denotes the identity sign, and $<\mathbf{a}, \mathbf{b}>$ represents the inner product between vectors $\mathbf{a}$ and $\mathbf{b}$.

%$\mathbb{R}^N$ and $\mathbb{R}^{N\times M}$ represent the $N$-dimensional vector space and $N$-by-$M$ real matrix space, respectively; $\mathbb{N}$ and $\mathbb{N}_{+}$ represents the set of non-negative and positive integers, respectively; $\mathbf{1}_n$ represents an $n$-dimensional vector where the components are all $1$'s; $\mathbbm{1}\{\cdot\}$ represents the indicative function with range $\{0,1\}$, $\leftarrow$ denotes the assignment from right to left, while $\rightarrow$ represents the approximation of the right term by the left term; $\triangleq$ defines the symbol on the left by equation on the right. in addition, $[\cdot]\triangleq\max\{\cdot,0\}$.
\section{System Models}
\label{sec:system_model}

As shown in Fig. \ref{fig:system_model}, we consider an unmanned aerial vehicle (UAV) flying in an urban outdoor environment within $T\in\mathbb{N}_+$ time steps and, as the receiver, constantly receiving wireless signals from a ground mobile transmitter (GMT).
The GMT is assumed to have access to its precise position at each time step, while the UAV is only equipped with an IMU that can roughly estimate its position.
Establish a fixed Cartesian coordinate system $X$-$Y$-$Z$, named as the world coordinate system, and let the real position of the UAV and the GMT at time step $t$ be $\mathbf{p}_{R,t}=[x_{R,t},y_{R,t},z_{R,t}]^T$ and $\mathbf{p}_{G,t}=[x_{G,t},y_{G,t},z_{G,t}]^T$, respectively, where $x_{R,t}, x_{G,t}$, $y_{R,t}, y_{G,t}$ and $z_{R,t}, z_{G,t}$ represent the $X$, $Y$ and $Z$ axis components of the UAV and the GMT, respectively, $t\in\{0,1,...,T\}$. Note that as the GMT can only move on the ground, its $Z$ axis component is fixed, i.e., $z_{G,t}\equiv h_G, \forall t$, where $h_G\in\mathbb{R}^+$ is a constant. 
The communication system between the GMT and the UAV works in the mmWave frequency band whose large bandwidth characteristic brings huge gains to the transmission rates. However, mmWave signals are easily attenuated in the atmosphere due to climate factors, which largely limits its coverage range \cite{mmwave}.  Hence, beamforming technique \cite{mmWaveBF, HBF} is used to realize the long-distance and high-quality communications between the GMT and the UAV. Moreover, to provide full angle coverage, the GMT and the UAV are both equipped with multiple arrays \cite{ny}. 
The GMT transmits its position along with other necessary information to the UAV through the wireless signals at each time step $t$ and moves to the next position every $T_c\in\mathbb{N}_+$ time steps. 
The UAV receives the signals from the GMT at each time step $t$ and moves to the next position in the next time step $t+1$. 
Let the signal channel be a standard multi-path cluster model \cite{cl}. 
Then, at each time step $t$, the wireless signals are propagated from the GMT to the UAV through several communication links (CLs). The CLs are determined by the environment and, hence,  contain rich information on the objects in the environment.
%that contain the information on the objects in the environment.  
%Therefore, the UAV flies around and aims to build the point cloud map of the surrounding environment and estimate its own trajectories simultaneously by mining the geometry information in the CLs. 
\subsection{Array and Signal Models}
Let the number of arrays at the UAV and the GMT be $N_\text{U}\in\mathbb{N}_+$ and $N_\text{G}\in\mathbb{N}_+$, respectively. Denote the angle of arrival (AOA) signals at the UAV side as $\Omega^u=(\phi^u,\theta^u)$, where $\phi^u$ and $\theta^u$ represent the azimuth and elevation angles, respectively. As in \cite{ny}, we define the spatial signature from the $j$-th array at the UAV as $\mathbf{v}^{(j)}_u(\Omega^u)$, $j\in\{1,2,...,N_U\}$.
The collection of all spatial signatures forms the array signature of the UAV, i.e.,  $\mathbf{v}_u(\Omega^u)=[\mathbf{v}^{(1)}_u(\Omega^u),...,\mathbf{v}^{(N_\text{U})}_u(\Omega^u)]$.
Note that  $\mathbf{v}^{(j)}_u(\Omega^u)$ includes both the array gain and element gain and is  normalized such that $||\mathbf{v}^{(j)}_u(\Omega^u)||^2$ is the directivity of the $j$-th array in the direction $\Omega^u$. 
Similarly, let the angle of departure (AOD) signals at the GMT side as $\Omega^g=(\phi^g,\theta^g)$, where $\phi^g$ and $\theta^g$ represent the azimuth and elevation angle, respectively. We define the spatial signature from the $k$-th array at the GMT as $\mathbf{v}^{(k)}_u(\Omega^u)$, $k\in\{1,2,...,N_U\}$. Then the collection of all spatial signatures forms the array signature of the GMT, i.e.,  $\mathbf{v}_g(\Omega^g)=[\mathbf{v}^{(1)}_g(\Omega^g),...,\mathbf{v}^{(N_\text{G})}_g(\Omega^g)]$.
Moreover, we define one-hot vectors $\tilde{\mathbf{v}}_u$ and $\tilde{\mathbf{v}}_g$ to represent the response on the arrays at the UAV and the GMT, respectively, i.e., $	\tilde{\mathbf{v}}_u(\Omega^u)=\big[0,...,0,\mathbf{v}^{(j')}_u(\Omega^u),0,...,0\big]$ and $	\tilde{\mathbf{v}}_g(\Omega^g)=\big[0,...,0,\mathbf{v}^{(k')}_g(\Omega^g),0,...,0\big]$, where $j'=\arg\max_{j}||\mathbf{v}^{(j)}_u(\Omega^u)||^2$,  and $k'=\arg\max_{k}||\mathbf{v}^{(k)}_g(\Omega^g)||^2$. Let $N_\text{G}^\text{dir}\in\mathbb{N}_+$ be the total number of beamforming directions at the GMT, and denote $\mathbf{w}^\text{g}_m$ as the beamforming vector for direction $m$ at the GMT, $m\in\{1,2,...,N^\text{dir}_\text{G}\}$. Note that $\mathbf{w}^\text{g}_m$ is called a \emph{codeword}, and the set of all codewords $\{\mathbf{w}^\text{g}_m, \forall m\}$ forms a \emph{codebook}. Similar to \cite{ny}, we design the codeword as 
$\mathbf{w}^\text{g}_m=\tilde{\mathbf{v}}_g(\tilde{\Omega}^g_m)/||\tilde{\mathbf{v}}_g(\tilde{\Omega}^g_m)||$, where $\tilde{\Omega}^g_m$ represents the departure angle in the beamforming under the codeword $\mathbf{w}^\text{g}_m$. The set of the departure angles in the beamforming then will be $\{\tilde{\Omega}^g_m, \forall m\}$. 
When the GMT conducts the beamforming of the codeword $\mathbf{w}^\text{g}_m$, the received signal of the UAV, $\mathbf{r}_m(t)$, can be represented as 
\begin{align}
	\mathbf{r}_m(t)=\sum_{l=1}^Lg_l\mathbf{v}_u(\Omega^u_l)\mathbf{v}_g(\Omega^g_l)^T\mathbf{w}^\text{g}_m{x}_m(t-\tau_l) + \mathbf{n}_m(t),
\end{align}
where $L$ is the total number of CLs, $g_l$ is the complex path gain of the $l$-th CL, $\Omega^u_l$ and $\Omega^g_l$ are the AOA and AOD of the $l$-th CL at the UAV and the GMT sides, respectively, $x_m(t)$ represents the complex
baseband synchronization signal transmitted in the $m$-th direction at the GMT side, $\tau_l$ is the time delay of the $l$-th CL, and $\mathbf{n}_m(t)$ represents the additive white Gaussian noise (AWGN). The CLs can be categorized into three types, including (1) the \emph{line-of-right} (LOS) CL where the signals propagate directly from the GMT to the UAV, (2) the first-order \emph{none line-of-sight} (NLOS) CL where the signals propagate from the GMT to the UAV through only one reflection, and (3) the high-order NLOS CL where more than one reflections occur during the propagations from the GMT to the UAV.

 \begin{figure*}
 	\centering
 	\includegraphics[width=170mm]{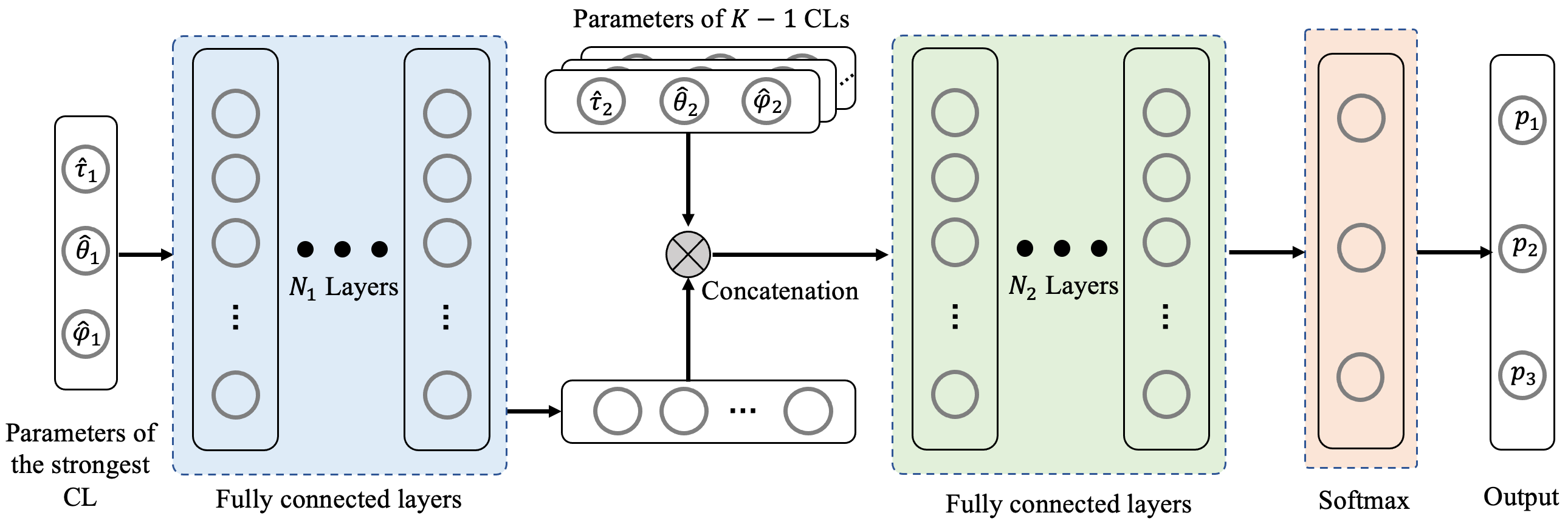}
 	\caption{The structure of the link state classification neural network (LSCN).}
 	\label{fig:classification_nn}
 \end{figure*}
\subsection{Problem Settings}
The goal of the considered SLAM problem is to build the 3D point cloud map $\mathcal{M}$ of the ambient urban outdoor environment at the UAV side during its flying based on the wireless communications with the GMT. 
As mapping needs the knowledge of the UAV's position at every time step, the UAV also needs to precisely locate itself during flying with the wireless signals and the aid of the IMU.
\section{Algorithms}
\label{sec:algorithm}
\begin{algorithm}[t]
	\normalsize\caption{{C-SLAM}}
	\label{algorithm:c-slam}
	\setstretch{1} %设置具有指定弹力的橡皮长度（原行宽的1.35倍）
	{\bf Inputs:} The position correction interval $T_c$, total time step $T$.\\
	{\bf Outputs:} The position estimations of the UAV, $\widehat{\mathbf{p}}_{R,t}, t\in\{0,1,...,T\}$ and the point cloud of the environment $\mathcal{M}$.\\
	{\bf Initialize:} The position of the GMT $\mathbf{p}_{G,0}$, the UAV corrects its position, i.e., $\widehat{\mathbf{p}}_{R,0}\leftarrow\widehat{\mathbf{p}}^\text{bs}_{R,0}$, 
	the empty point cloud $\mathcal{M}$, 
	the counter $n\leftarrow 0$.
	\begin{algorithmic}[1]
		\normalsize 
		\For{$t=0$ to $T$}
		\State The GMT transmits communication signals to the UAV, including its position $\mathbf{p}_{G,t}$.
		\State The UAV utilizes LTD to obtain path estimations of the top $K$ strongest CLs, $\{(\widehat{\tau}_l,\widehat{\theta}_l, \widehat{\phi}_l, \gamma_l), l=1,...,K\}$.
		\State The UAV classifies the path state with the LSCN.
		\If {the path state is the first-order NLOS}
		\State The UAV calculates the coordinate of the reflection point $[x_P, y_P, z_P]^T$ with \eqref{equ:geores_1}, \eqref{equ:geores_2} and \eqref{equ:geores_3}.
		\State $\mathcal{M}\leftarrow\mathcal{M}\cup\{[x_P, y_P, z_P]^T\}$.
		\EndIf
		\State The UAV moves to the next position $\mathbf{p}_{R,t+1}$.
		\If{$(t+1)\%T_c =0$}
		\State $n \leftarrow n + 1$.
		\State The GMT moves to the next position $\mathbf{p}_{G,t+1}$.
		\State The UAV corrects its current position estimation with the BSM, i.e., $\widehat{\mathbf{p}}_{R,t+1}\leftarrow\widehat{\mathbf{p}}^\text{bs}_{R,n}$.
		\Else
		\State 	The UAV uses the IMU to estimate its current position, i.e.,  $\widehat{\mathbf{p}}_{R,t+1}\leftarrow \widehat{\mathbf{p}}^\text{bs}_{R,n}+\sum_{q=(n-1)T_c}^{t}(\widehat{\mathbf{p}}^\text{imu}_{R,q+1}-\widehat{\mathbf{p}}^\text{imu}_{R,q})$.
		\EndIf
		\EndFor
	\end{algorithmic}
\end{algorithm}
The C-SLAM framework is composed of two algorithms, including a hybrid periodic positioning calibration (HPC) algorithm for localization, and the DLM for mapping. 
\subsection{Hybrid Periodic Position Calibration Algorithm}
We combine the IMU and the BSM \cite{bs_l} for localization of the UAV, which forms the HPC algorithm.
Notably, the BSM can help to correct the accumulative error on the position estimations caused by the IMU. Hence, the HPC can realize real-time and precise localization with only the communication signals and the aid of the IMU. We next briefly introduce the working principles of the IMU and the BSM, respectively.
\begin{itemize}
	\item IMU: The main components of the IMU are gyroscope, accelerometer and magnetometer. The gyroscope can obtain the acceleration of each axis, the accelerometer can obtain the acceleration in the $X$, $Y$, and $Z$ directions, and the magnetometer can obtain the information of the surrounding magnetic field. The IMU fuses the data 
	of the these three sensors to calculate the positions.
	\item BSM: In wideband communications, the beamforming of the subcarriers may not point to the target position. Such a phenomenon is named as the \emph{beam squint} as the beamforming direction gradually “squint” over the frequency. Note that
	with the aid of the TDs, the range and trajectory of the beam
	squint can be freely controlled, and hence it is possible to
	reversely utilize the beam squint for localizations \cite{bs_l}. The BSM designs a way to control the trajectory of the beam
	squint points. With the design, beamforming from
	different subcarriers would purposely point to different angles
	and different distances such that the UAV from different positions
	would receive the maximum power at different subcarriers.
	Hence, the position of the UAV can be determined from
	the beam squint effect.
\end{itemize}
The HPC works in an periodic manner, correcting the estimations of the IMU with the BSM every $T_c\in\mathbb{N}_+$ time steps. Specifically, let the estimated position of the UAV generated by the IMU at time step $t$ be $\widehat{\mathbf{p}}^\text{imu}_{R,t}$,  and let the $n$-th estimation on the position of the UAV generated by the BSM be $\widehat{\mathbf{p}}^\text{bs}_{R,n}$, $n\in \mathbb{N}$. Then, when the time step $t$ is exactly an integer multiple of $T_c$, i.e., $t \% T_c =0$, the UAV corrects its current position estimation with the BSM, i.e.,
\begin{align}
	\label{equ:bsm}
	\widehat{\mathbf{p}}_{R,t}\leftarrow\widehat{\mathbf{p}}^\text{bs}_{R,n},
\end{align}
where $n = t/T_c$. When the time step $t$ is not exactly an integer multiple of $T_c$, i.e., $t \% T_c \neq 0$, the UAV estimates its position based on the position variation calculated by the IMU from the latest position estimation given by the BSM, i.e., 
\begin{align}
	\label{equ:imu}
	\widehat{\mathbf{p}}_{R,t}\leftarrow \underbrace{\widehat{\mathbf{p}}^\text{bs}_{R,n}}_{\text{latest correction} }+\underbrace{\sum_{q=(n-1)T_c}^{t-1}(\widehat{\mathbf{p}}^\text{imu}_{R,q+1}-\widehat{\mathbf{p}}^\text{imu}_{R,q})}_{\text{position variations given by the IMU}}.
\end{align}

\subsection{Deep Learning Based Mapping Algorithm}
Our basic idea of mapping is to find the coordinates of the reflection points of the CLs on the surface of objects in the environment and use these reflection points to generate the 3D point cloud map of the environment.
Recall that there are three types of the CL, including the LOS CL, the first-order NLOS CL and the higher-order NLOS CL. The LOS CL has no interactions with the environment, and hence, cannot be utilized to generate the point clouds. Although the higher-order NLOS CLs have multiple reflection points, we prove that their coordinates cannot be solved theoretically.
\begin{proposition}
	The coordinates of reflections points on the higher-order NLOS CL have infinite possible solutions and cannot be uniquely determined.
\end{proposition}
\begin{proof}
	See Appendix \ref{app:proof_pro_1}.
\end{proof}
Therefore, we can only leverage the reflection points on the first-order NLOS CL to generate the 3D point cloud map.
However, before calculating the reflection points, we need to identify the first-order NLOS CL from all CLs. 
Note that identifying the types of all CLs is pretty difficult, if not impossible. Hence, we here manage to classify the type of \emph{the strongest CL} that is defined below instead of all CLs. 
\begin{definition}
	Sort the CLs by their signal noise ratio (SNR) measured at the UAV side from large to small, and we define the $l$-th strongest CL as the $l$-th CL after sorting. The strongest CL refers to the $1$-st strongest CL.
\end{definition}
Specifically, we first adopt the path decomposition method to find the \emph{channel state information} (CSI) of all CLs and design the LSCN to classify the CLs based on the CSI. Then we present the analytical results of the geometry calculation on the coordinates of the  reflection points on the first-order NLOS CL.
The path decomposition method, LSCN and the calculations of the reflection points coordinates construct the whole DLM algorithm and are described as follows, respectively.

\subsubsection{Path Decomposition Method}
Path decomposition methods have been studied in many literatures \cite{ny, pe_1, pe_2}. We here adopt a low-rank tensor decomposition (LTD) algorithm \cite{ny} to estimate the CSI of all the CLs. Specifically, define a spatial and temporal correlation factor
\begin{align}
	\rho_m(\tau,\Omega^u)\triangleq\int \tilde{\mathbf{v}}_{u}(\Omega^{u})*\mathbf{r}_m(t)x^*_m(t-\tau)dt,
\end{align}
where $\tau$ denote the time delay. 
The magnitude of $\rho_m(\tau,\Omega^u)$ has peaks when the input values are exactly the parameters of the $l$-th CL, i.e., $\tau=\tau_l$ and $\Omega^u=\Omega_l^u, \forall l$.
The LTD algorithm can find these peaks, and thereby, determine the estimations of the CL parameters $\{\widehat{\tau}_l, \widehat{\Omega}^u_l \mid  \forall l\}$, where $\widehat{\tau}_l$, $\widehat{\Omega}^u_l$ and $\widehat{\gamma}_l$ denote the estimated time delay and AOA of the $l$-th CL. Note that the estimated AOA includes the estimations of the azimuth and elevation angles, i.e., $\widehat{\Omega}^u_l=(\widehat{\theta}^u_l,\widehat{\phi}^u_l), \forall l$. Moreover, we can also determine the signal noise ratio (SNR) of each CL, denoted as  $\gamma_l, \forall l$.

\begin{figure}
	\centering
	\includegraphics[width=80mm]{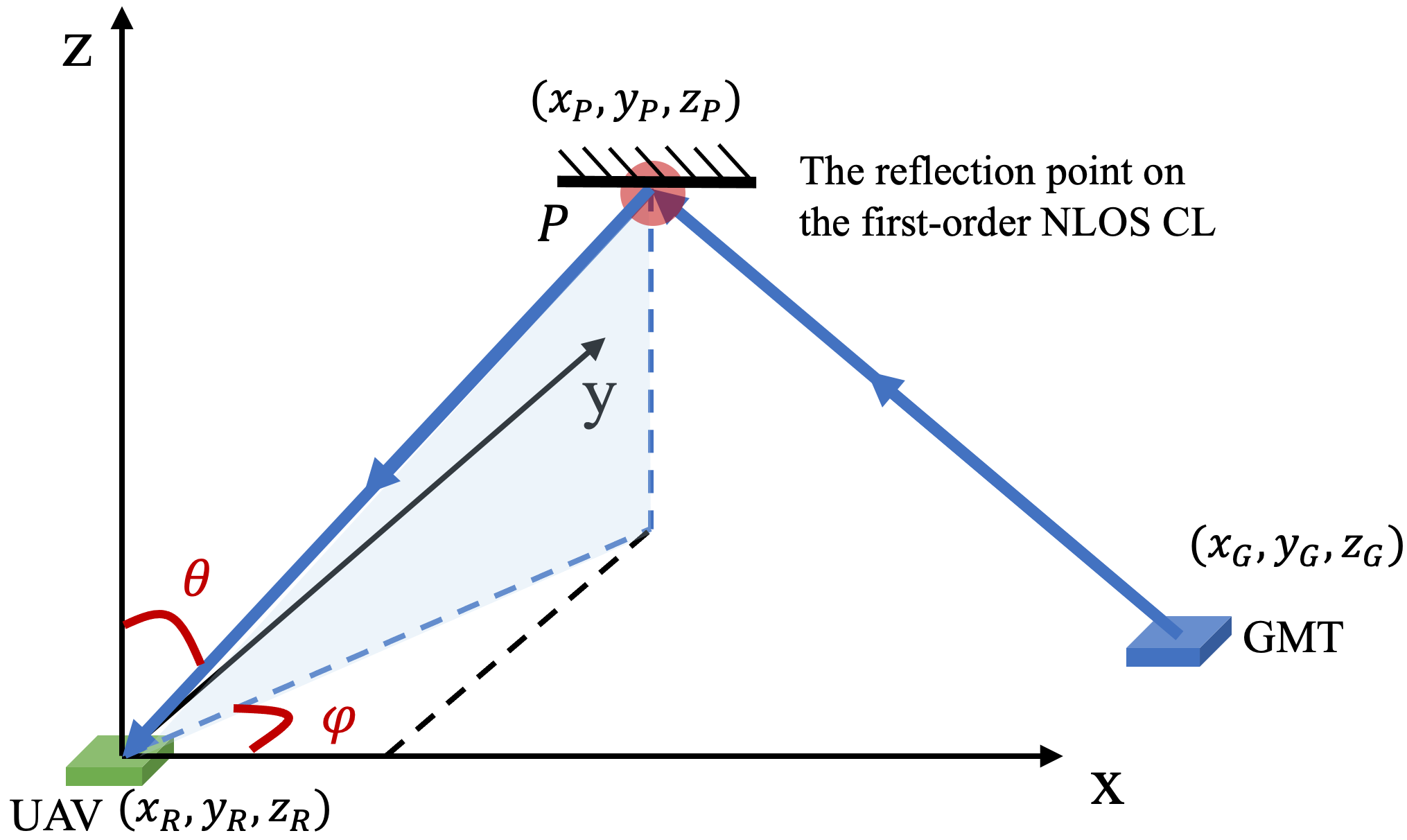}
	\caption{The geometry analysis of the first reflection path.}
	\label{fig:geometry}
\end{figure}
\subsubsection{Link State Classification Neural Network (LSCN)}
{At each time step, there exist several different types of CLs between the GMT and the UAV. Notably, classifying the type of the strongest CL is feasible.  The reasons are as follows. Label the LOS CL, the first-order NLOS CL and the higher-order NLOS CL as $1$, $2$, and $3$. We claim that the strongest CL cannot have a higher label than all other CLs  at each time step.
Hence, the strongest CL can represent the main features of the relative positions between the GMT and the UAV. For example, if the strongest CL is the first-order NLOS CL with label $2$, then there is no LOS CL with label $1$ in the remaining CLs, and we know there are obstructs in the line between the GMT and the UAV. However, other CLs do not have such functions since there might be other CLs with type of smaller labels. Besides, the main features of the relative positions between the GMT and the UAV can be learned by training with large data generated in the environment. Therefore, identifying the type of the strongest CL is feasible. For convenience, we define the following three \emph{link states}: 
	\begin{itemize}
		\item \emph{LOS:} the strongest CL is LOS CL;
		\item \emph{first-order NLOS:} the strongest CL is the first-order NLOS CL;
		\item \emph{higher-order NLOS:} the strongest CL is the higher-order NLOS CL.
	\end{itemize}
To classify the link state at each time step, we design a LSCN as shown in Fig. \ref{fig:classification_nn}.
The LSCN takes the estimated parameters of the top $K\in\mathbb{N}_+$ strongest CLs as inputs and generates the probabilities of being three link states. Specifically, the estimated parameters of the strongest CL $(\widehat{\tau}_1, \widehat{\theta}_1, \widehat{\phi}_1)$ are processed by $N_1\in\mathbb{N}_+$ fully connected layers. The resulted vector is concatenated with the estimated parameters of the remaining $K-1$ CLs and processed by $N_2\in\mathbb{N}_+$ fully connected layers. 
Note that we design the number of neurons in the last layer of the second fully connected neural network as $3$ since there are three types of link states.}
For convenience, we denote the functions of the neural networks with $N_1$ and $N_2$ fully connected layers as $F_1(\cdot)$ and $F_2(\cdot)$, respectively.
Then the resulted vector $\mathbf{y}\in\mathbb{R}^3$ can be expressed as 
\begin{align}
	\mathbf{y}=F_2\bigg(\bigg\{\widehat{\tau}_l,\widehat{\theta}_l, \widehat{\phi}_l \mid l=2,...,K\bigg\} \| F_1(\widehat{\tau}_1,\widehat{\theta}_1, \widehat{\phi}_1); \mathbf{\Psi}\bigg),
\end{align}
where $\mathbf{\Psi}$ denotes all the trainable parameters in $F_1(\cdot)$ and $F_2(\cdot)$. The resulted vector $\mathbf{y}$ is processed by the softmax function, and the output of the LSCN is $\widehat{\mathbf{p}} \triangleq [\widehat{p}_1,\widehat{p}_2,\widehat{p}_3]=\text{softmax}(\mathbf{y})$,
where $\widehat{p}_1$, $\widehat{p}_2$ and $\widehat{p}_3$ represent the probability of the link state being the LOS, first-order NLOS and the higher-order NLOS, respectively. To train the LSCN, we build a training set $\mathcal{D}=\{ \{(\widehat{\tau}^{d}_l, \widehat{\theta}^{d}_l, \widehat{\phi}^{d}_l, \mathbf{p}^{d}),
\mid l=1,2,...,K
\} \mid d=1,2..., |\mathcal{D}|\}$ in the considered urban outdoor environment, where $\mathbf{p}^{d}$ is the one-hot vector representing the ground truth of the link state. 
Note that, as in \cite{ny}, we normalize all the input values to the scale of $[-1,1]$ for ease of training convergence. We leverage the cross entropy function as the loss function $\mathcal{L}$ of the LSCN, i.e., 
\begin{align}
	\mathcal{L}(\mathbf{\Psi})=-(\mathbf{p}^d)^T\log (\widehat{\mathbf{p}}^d),
\end{align}
where $\widehat{\mathbf{p}}^d$ denotes the output of the LSCN with the input of the $d$-th training data.
The LSCN is trained with the stochastic gradient descent, i.e.,
\begin{align}
	\mathbf{\Psi} \leftarrow \mathbf{\Psi} -\alpha\nabla_\mathbf{\Psi}\mathcal{L}(\mathbf{\Psi}).
\end{align}
where $\alpha \in\mathbb{R}^+$ is the learning rate.

\subsubsection{Analytical Expression of 3D Point Cloud Coordinates}
At the time step when the link state is the first-order NLOS, we can analyze the geometry relationship of the strongest CL and calculate the reflection point on it as the cloud point.

\begin{strip}
	\hrule
	\begin{align}
		\centering
		\label{equ:geores_1}
		&		x_P=\frac{c^2\tau^2 + x_R^2 - x_G^2-(y_R-y_G)^2-(z_R-z_G)^2 + 2\big[(y_R-y_G)\tan({\phi})\;\textcolor{red}{\pm}\;(z_R-z_G)\cot({\theta})\sqrt{1+\tan^2({\phi})}\;\textcolor{red}{\pm}\; c\tau A\big]x_R
		}
		{2\big[\;\textcolor{red}{\pm}\; c\tau A+x_R-x_G\;\textcolor{red}{\pm}\; (z_R-z_G)\cot({\theta})\sqrt{1+\tan^2({\phi})}+(y_R-y_G)\tan(\phi)\big]},\\ 
		\label{equ:geores_2}
		&\qquad\qquad\qquad\qquad\qquad\qquad\qquad\qquad\qquad	y_P=y_R+\tan({\phi})(x_P-x_R), \\
		\label{equ:geores_3}
		&\qquad\quad\qquad\qquad\qquad\qquad\qquad\qquad\; z_P=z_R\;\textcolor{red}{\pm}\;\cot({\theta}) (x_P-x_R)\sqrt{1+\tan^2({\phi})},
	\end{align}
	\hrule
\end{strip}
For convenience, we omit the subscript $t$ in the position of the GMT and the UAV and directly use $\phi$, $\theta$ and $\tau$ to denote the estimated parameters of the strongest CL, $\widehat{\phi}_1$, $\widehat{\theta}_1$ and $\widehat{\tau}_1$ in the following calculations.
As shown in \reffig{fig:geometry}, the coordinate of the reflection point $[x_P, y_P, z_P]^T$ on the first-order NLOS CL satisfies
\begin{align}
	\label{equ:geo}
	\left\{
	\begin{aligned}
	&\frac{y_P-y_R}{x_P-x_R}=\tan ({\phi}),\\
	&\frac{z_P-z_R}{\sqrt{(x_P-x_R)^2+(y_P-y_R)^2}}=\tan(\frac{\pi}{2}-{\theta}),\\
	&\left\|[x_P, y_P,z_P]^T-[x_R, y_R, z_R]^T\right\|_2=c{\tau}-\\
	&\left\|[x_P, y_P,z_P]^T-[x_G, y_G, z_G]^T\right\|_2.
	\end{aligned}
	\right.
\end{align}
Recall that the GMT sends its position to the UAV at each time step, which makes $x_G$, $y_G$ and $z_G$ known variables. Besides, as the UAV locates itself and carries out path estimations from time to time, the position of UAV $[x_R,y_R, z_R]^T$ and the AOA ${\theta}$, ${\phi}$ and the time delay ${\tau}$ are all known variables. 
We derive the closed-form expressions of $x_P$, $y_P$ and $z_P$ based on these known variables as follows.
\begin{proposition}
	The closed-form expressions of the coordinate $[x_P,y_P,z_P]$ of the reflection point on the first-order NLOS CL are \eqref{equ:geores_1}, \eqref{equ:geores_2} and \eqref{equ:geores_3}, respectively, where $A=\sqrt{(1+\tan^2({\phi}))(1+\cot^2({\theta}))}$. 
\end{proposition}
\begin{proof}
	See Appendix \ref{app:proof}.
\end{proof}
Note that in \eqref{equ:geores_1} and \eqref{equ:geores_3}, all $\textcolor{red}{\pm}$'s take the plus or minus sign synchronously. This means that there are two solutions to \eqref{equ:geo}. Specifically, when $x_P\ge x_R$, it takes the plus sign, and when $x_P< x_R$, it takes the minus sign. 
This gives us a way to filter the false solution, i.e.,  $\textcolor{red}{\pm}$ takes the plus sign if ${\phi}\in[-\frac{\pi}{2}, \frac{\pi}{2}]$ and takes the minus sign otherwise.
The UAV constantly calculate the coordinates of the reflection points on the strongest CL in the first-order NLOS path state, and the point clouds of the environment can be gradually built. Note that there is no distortion problem in generating the point clouds since the calculations are carried out in the world coordinate system.

\subsection{Overall C-SLAM Algorithm}
The overall C-SLAM algorithm is summarized in Algorithm \ref{algorithm:c-slam}. Note that the localization and mapping are carried out in the UAV simultaneously. Specifically, the UAV locates its position with GPS and uses the for correction every $T_c$ time steps. At each time step, the GMT transmits its position to the UAV though wireless signals, and the UAV carries out the DLM to build the point cloud of the environment.

\section{Simulation Results}
\label{sec:simulations}
\begin{figure}[t]
	\centering
	\includegraphics[width=80mm]{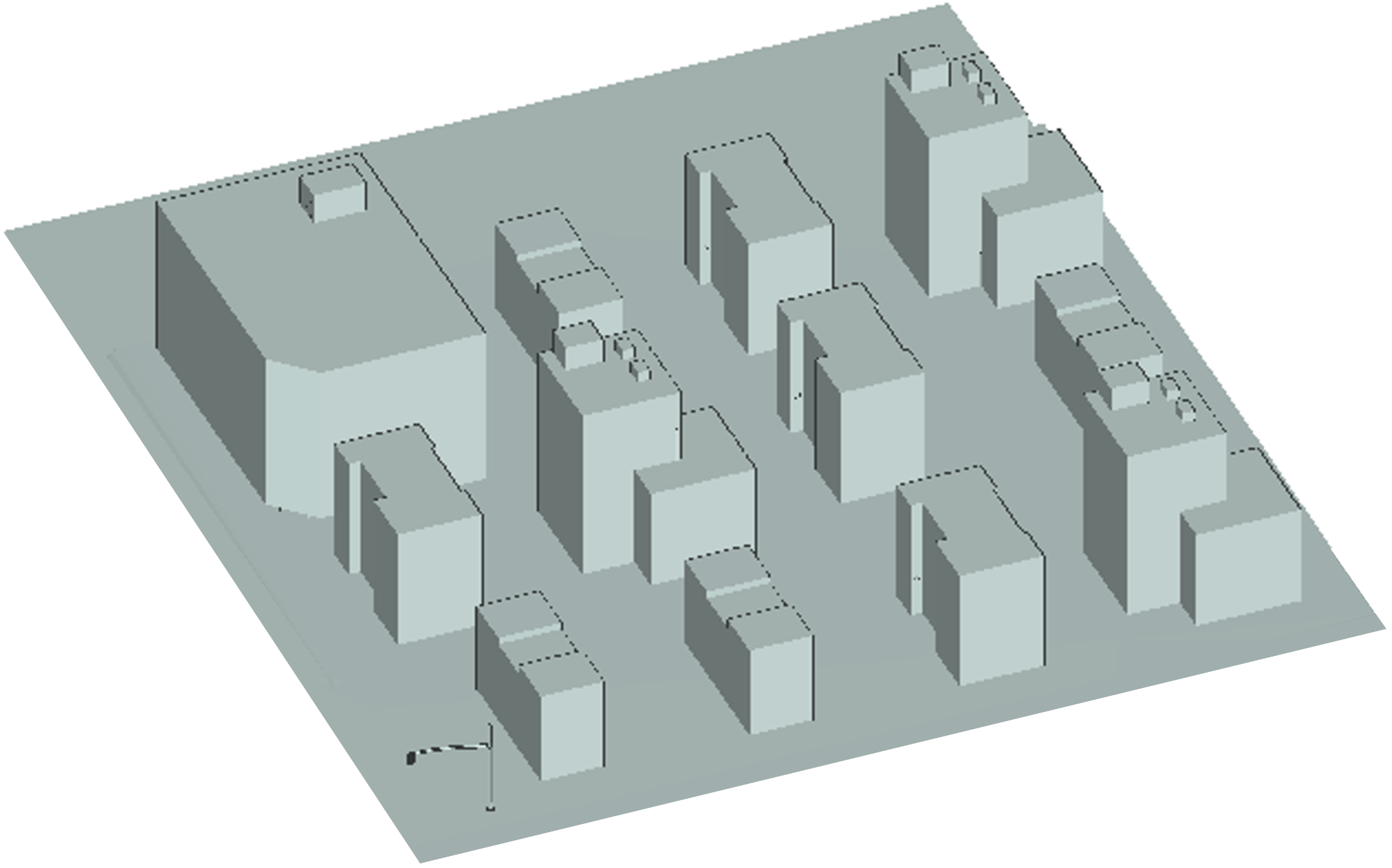}
	\caption{The considered $120$m$\times 120$m outdoor environment, consisting of a ground with multiple irregular buildings.}
	\label{fig:building}
\end{figure}

\begin{figure}[t]
	\centering
	\includegraphics[width=80mm]{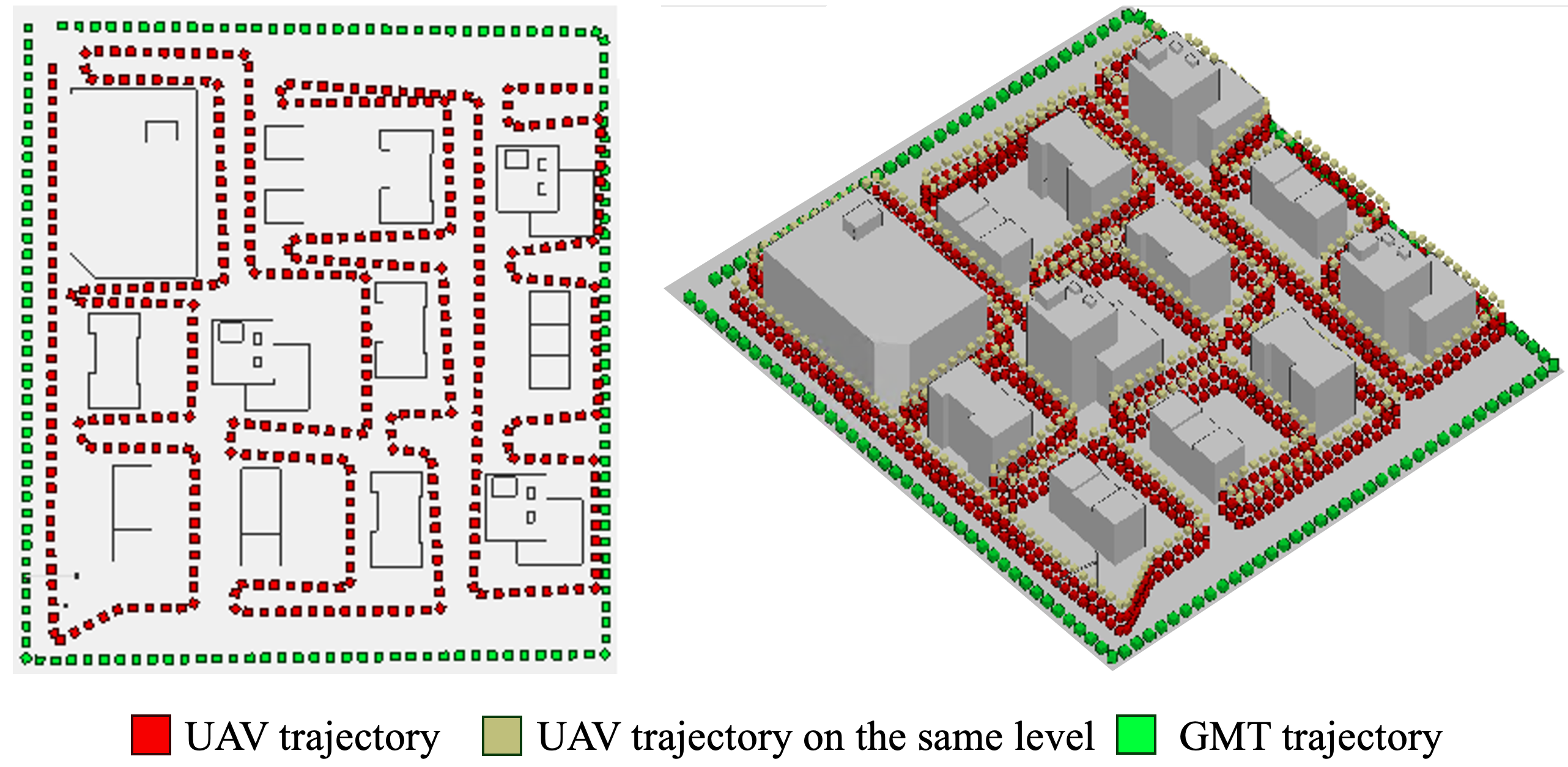}
	\caption{The 2D and 3D trajectories of the UAV and the GMT, where the GMT is moving on the ground around the buildings, and the UAV is flying between the buildings on different levels. }
	\label{fig:trajectory}
\end{figure}

\begin{figure}[t]
	\centering
	\includegraphics[width=80mm]{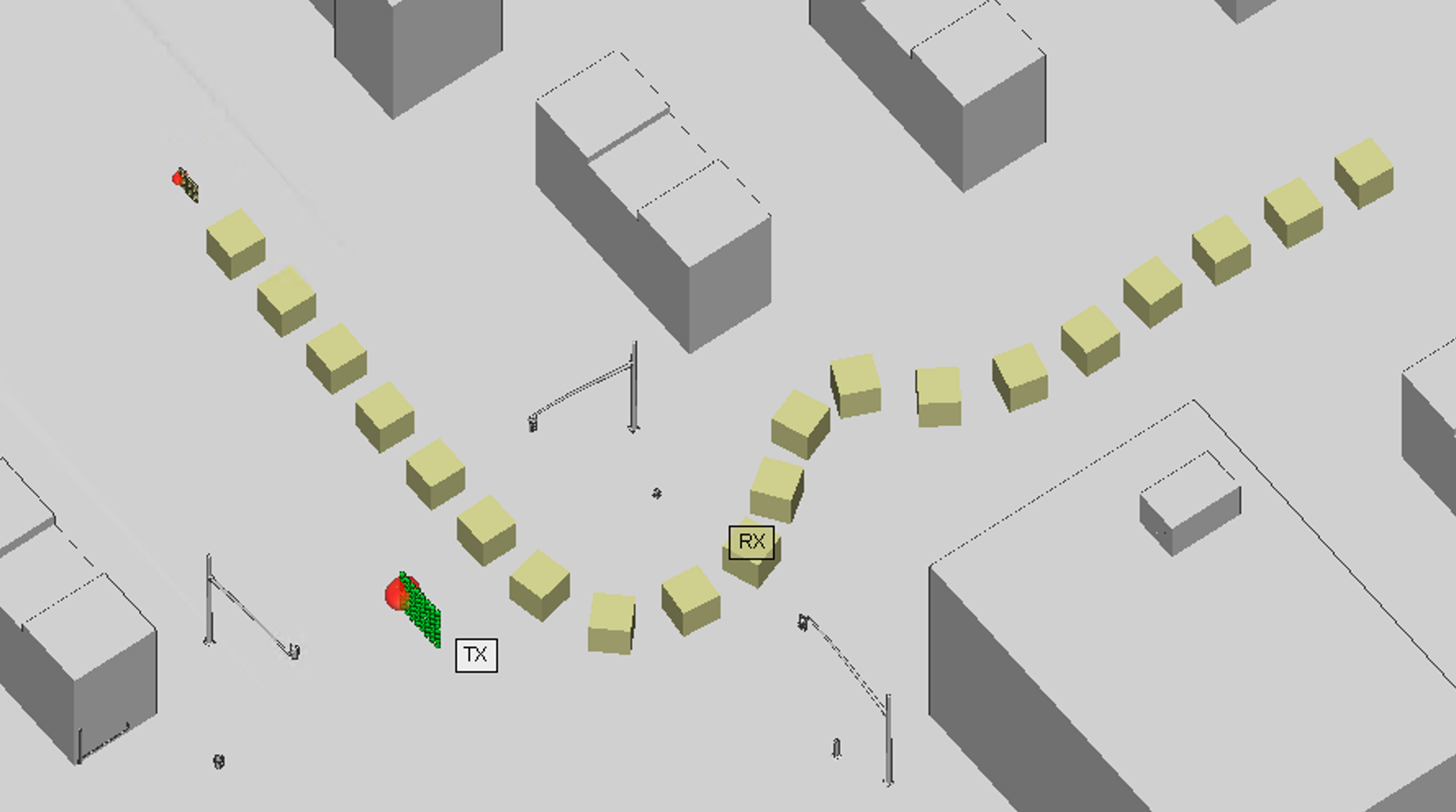}
	\caption{The arrays of the GMT and the UAV.}
	\label{fig:mimo}
\end{figure}
We consider a $120$m$\times 120$m$\times 20$m outdoor environment consisting of a ground and multiple irregular 3D buildings, as shown in Fig. \ref{fig:building}. The model of the considered environment is built in Blender, an open-source 3D computer graphics software tool set.
As shown in Fig. \ref{fig:trajectory}, the GMT is moving on the ground around the buildings, and the UAV is flying between the buildings on different levels. 
The ground mobile transmitter (GMT) and the UAV are each equipped with a $30$GHz antenna array, forming a mmWave communication system.
The parameters for these two arrays are summarized in Table. \ref{table:array}. 
Note that to balance the array gain and the hardware implementation cost, we adopt the hybrid beamforming technique \cite{HBF} for both arrays.
{Fig. shows the beam patterns of the arrays. To generate the wireless signals in the considered environment, we leverage the Wireless InSite by Remcom \cite{remcom}. Wireless InSite is a professional suite of RF propagation models, providing 3D ray-tracing and fast ray-based methods, that has been widely used in academia and industry.}
Without loss of generality, we neglect the difference of the materials of the surface of the ground and all buildings and model them all as one-layer dielectric with $15 F/m$ permittivity and $0.015 S/m$ conductivity. The calculation of the ray-tracing is accelerated by the NVIDIA GeForce RTX 2080 SUPER. 
\begin{table}[h]
	\centering
	\caption{Parameters of the antenna arrays on the GMT and the UAV.}
	\label{table:array}
	\begin{tabular}{c|c|c}
		%	\specialrule{0em}{1pt}{1pt}
		\hline
		\rowcolor[gray]{0.9}
		\small\textbf{{Parameters}}&\small\textbf{{GMT}}&\small\textbf{{UAV}}\\
		\hline
		\small Carrier Frequency&\multicolumn{2}{c}{\small $30$GHz}\\
		\hline
		\small Bandwidth&\multicolumn{2}{c}{\small $200$MHz}\\
		\hline
		\small Antenna Element Type & \multicolumn{2}{c}{\small half-wave dipole}\\
		\hline
		\small Waveform & \multicolumn{2}{c}{\small Sinusoid}\\
		\hline
		\small Array Size & \small $64$ ($8 \times 8$ UPA) &\small $8$ ($1\times 8$ ULA)\\
		\hline
			\small Array Number& \small $3$ &\small  $4$\\
		\hline 
		\small Transmit Power & \small $30$ dBm &\small  $-$\\
		\hline 
		
	\end{tabular}
	% \vspace{-0.8cm}
\end{table}
\begin{figure}[t]
	\centering
	\includegraphics[width=80mm]{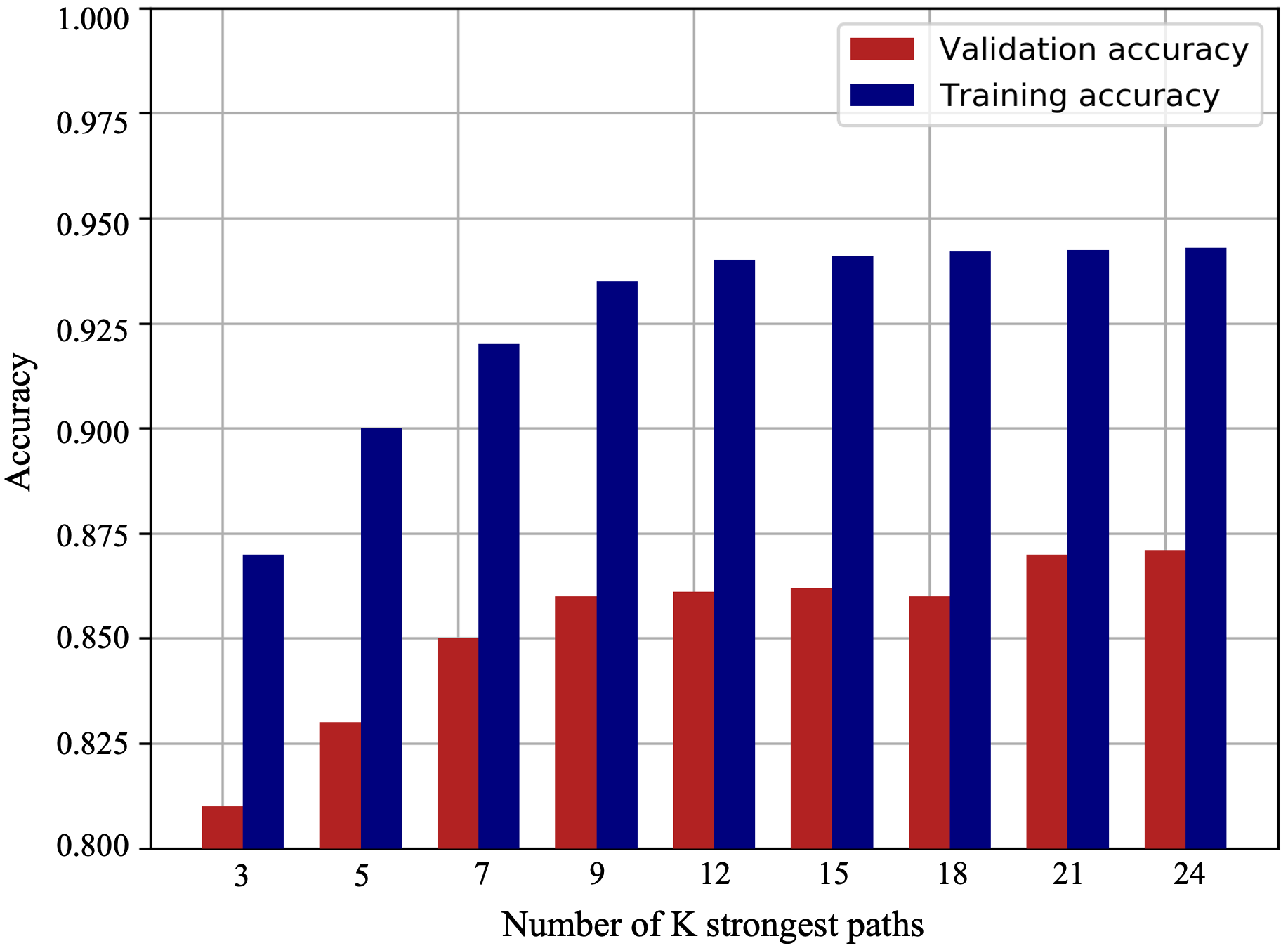}
	\caption{The training and validation accuracy verse the number of strongest paths $K$.}
	\label{fig:K}
	
\end{figure}
\begin{figure}
	\centering
	\includegraphics[width=75mm]{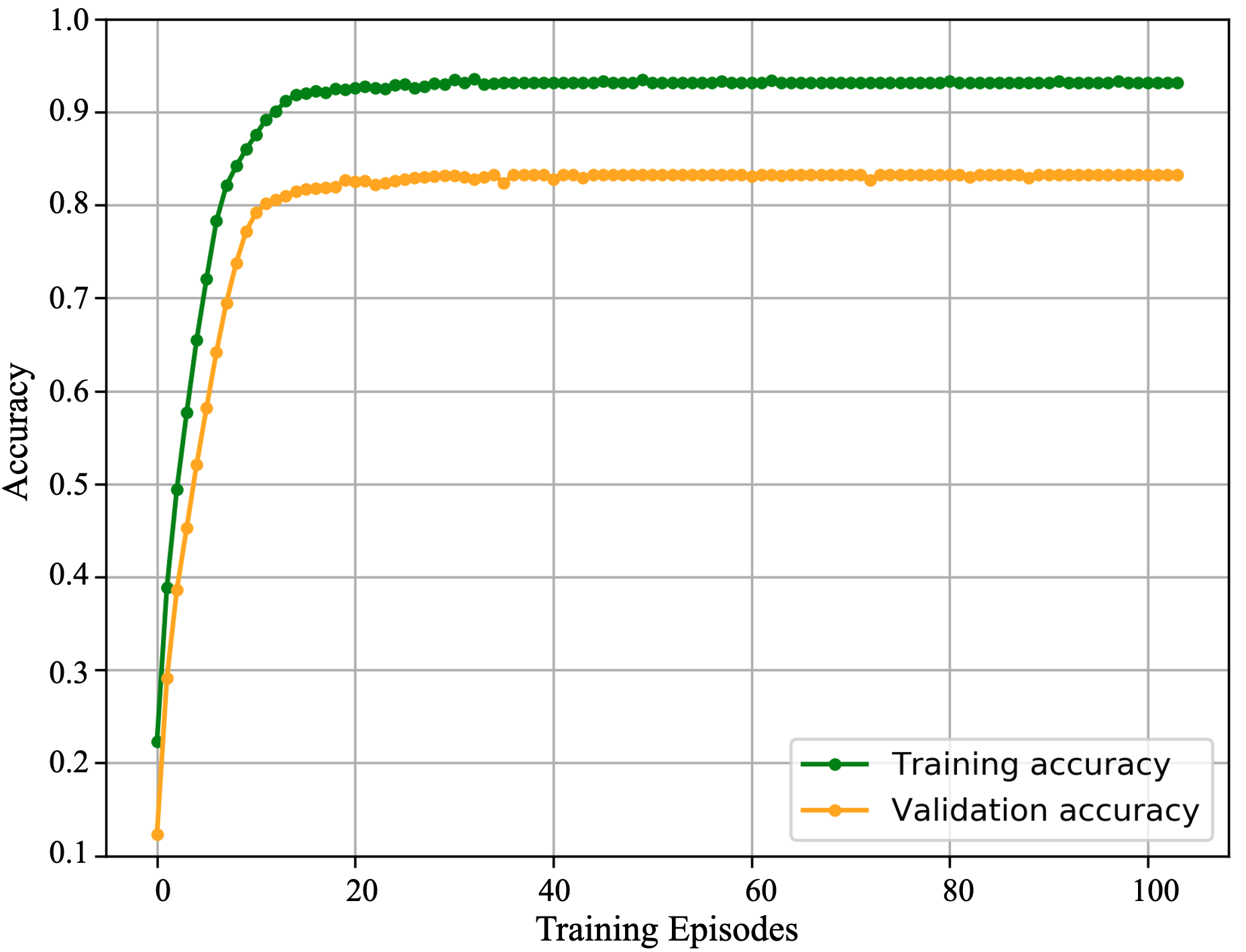}
	\caption{The training and validation accuracy curves when $K=9$.}
	\label{fig:accuracy_curve}
	
\end{figure}
\begin{figure}
	\centering
	\includegraphics[width=80mm]{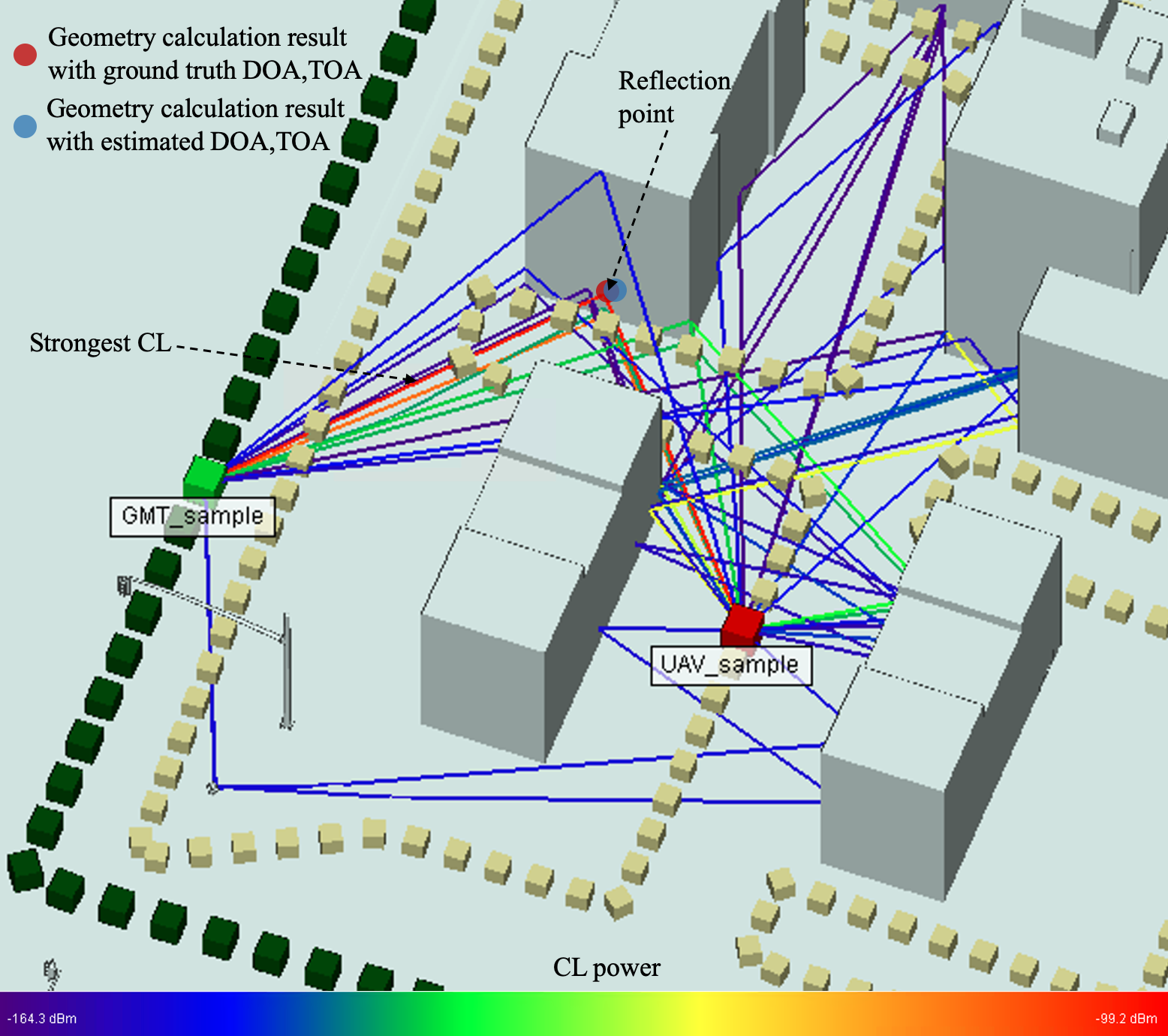}
	\caption{The multi-path propagation between the GMT and the UAV in the considered outdoor environment. The color of each CL represents its power, where the redder the color of the CL, the stronger the power of it. Here the link state is a first-order NLOS CL}
	\label{fig:reflections}
\end{figure}

We conduct extensive experiments to validate the effectiveness of our approach. Specifically, we focus on the following three questions in our experiments: (1) How does the link-state classification neural network perform? (2) What is the performance and time complexity of the geometric point cloud generation algorithm? (3) What is the performance of the overall C-SLAM algorithm?
\subsection{Link-state Classification Results}
\emph{To answer question (1):} We construct the training and validation set for the LSCN in the considered outdoor environment. Specifically, we place $120\times 120$ uniform X-Y grid RXs that have the same array with the UAV in $z=1\text{m}, 2\text{m},...,20\text{m}$ plane, which makes up a total of $288,000$ RX's. Note that each RX is $1$m apart from the adjacent RXs.  We place 
$100$ TXs that have the same array with the GMT in $100$ different locations. The resulting communication links between the TXs and RXs makes up the training set for the LSCN. Similarly, we place $50$ TXs that have the same array with the GMT in other different $50$ locations to generate the validation set. Recall that the input of the LSCN is the path estimations of $K$ strongest paths, which makes the input dimension be $3K$. We design the LSCN as a fully connected neural network with three hidden layers, where the number of neuron units are $10,50,100$, respectively. The output layer has $3$ neuron units processed by the  softmax function, which represent the probability of being LOS, 1st-order NLOS and higher-order NLOS. The LSCN is implemented by pyTorch, and its hyper-parameters are summarized in Table. \ref{table:LSCN}. The training and validation accuracies under different values of $K$ are shown in Fig. \ref{fig:K}. We can see that the training and validation accuracies both increase with the number of $K$. This is consistent with the common sense since larger $K$ involves more paths and brings more information to the LSCN. However, the growth rates of the accuracies largely slow down when $K>9$. Hence, to balance between the accuracy and the time complexity, we choose $K$ to be $9$. When $K=9$, the loss and accuracy curves during training are shown in Fig. \ref{fig:accuracy_curve}. We can see that the classification accuracies on the training and validation set both rises with the training episodes and converge to $0.935$ and $0.830$, respectively.

\begin{table}[h]
	\centering
	\caption{Hyper-parameters of the LSCN.}
	\label{table:LSCN}
	\begin{tabular}{c|c|c|c}
		%	\specialrule{0em}{1pt}{1pt}
		\hline
		\rowcolor[gray]{0.9}
		\makecell[c]{\small\textbf{{Hyper-}}\\\textbf{{parameters}}}
	&\small\textbf{{LSCN}}&		\makecell[c]{\small\textbf{{Hyper-}}\\\textbf{{parameters}}}&\small\textbf{{LSCN}}\\
		\hline
		\small Input units& $3K$&	\small Optimizer& Adam\\
		\hline
		\small Hidden layer& $[10, 50, 100]$&	\small Learning rate& $\sim 0.001$\\
	\hline
		\small Output units& $3$&	\small Batch size& $1000$\\
	\hline
		
	\end{tabular}
	% \vspace{-0.8cm}
\end{table}

\subsection{3D Point Clouds Generation Results}

\emph{To answer question (2):} We sample $10,000$ first-order path state and calculate the coordinates of the reflection points with the geometric point cloud generation algorithm. Notably, we here conduct two experiments. In the first experiment, we estimate the AOA and TOA of the strongest CLs using the LTD method. In the second experiment, we directly use the AOA and TOA given by the Wireless InSite software, which can be viewed as the ground truth. The mean square errors (MSE) between the coordinates of the generated point clouds and the coordinates of the real reflection points are shown in Table. \ref{table:MSE}. We can see that the MSE in the second experiment is nearly zero, smaller than that in the first experiment. This means that the geometric solutions in \eqref{equ:geores_1}, \eqref{equ:geores_2} and \eqref{equ:geores_3} are correct, and the point cloud generation error comes from the error of path estimations. Nonetheless, the MSE in the first experiment is about $0.2$m, which is rather small. Fig. \ref{fig:reflections} shows the top $25$ strongest CLs between the GMT and the UAV when they are at positions $[28.20, 23.04, 2]^T$ and $[53.97, 23.24,2]^T$, respectively. We can see that the path state is the first-order NLOS since there is no direct CLs from the GMT to the UAV and the strongest CL has only one reflection. The real coordinate of the reflection point of the strongest CL is $[41.59, 39.09, 2]^T$, and the geometry calculation results with ground truth and estimated AOA and TOA are  $[41.59, 39.09, 2]^T$ and $[41.68, 39.18, 2.01]^T$, respectively. We can see that the former result is nearly identical to the coordinate of the real reflection point and the latter result only has $0.12$m away from it, which validates the effectiveness of the geometric point cloud generation algorithm.

\begin{table}[h]
	\centering
	\caption{MSE between the coordinates of the generated point clouds and the coordinates of the real reflection points.}
	\label{table:MSE}
	\begin{tabular}{c|c|c}
		\specialrule{0em}{1pt}{2pt}
		\hline
		\small Experiments&\small\makecell[c]{{Experiment 1:}\\{Estimated AOA and TOA}\\{by the LTD}}& 	\small\makecell[c]{{Experiment 2:}\\{Using ground truth}\\ {AOA and TOA}}\\
		\hline
		\small{{MSE} (m)}&$0.2112$&$4.65\times 10^{-3}$\\
		\hline

	\end{tabular}
	% \vspace{-0.8cm}
\end{table}

\subsection{Overall Results}

\begin{figure}
	\centering
	\includegraphics[width=80mm]{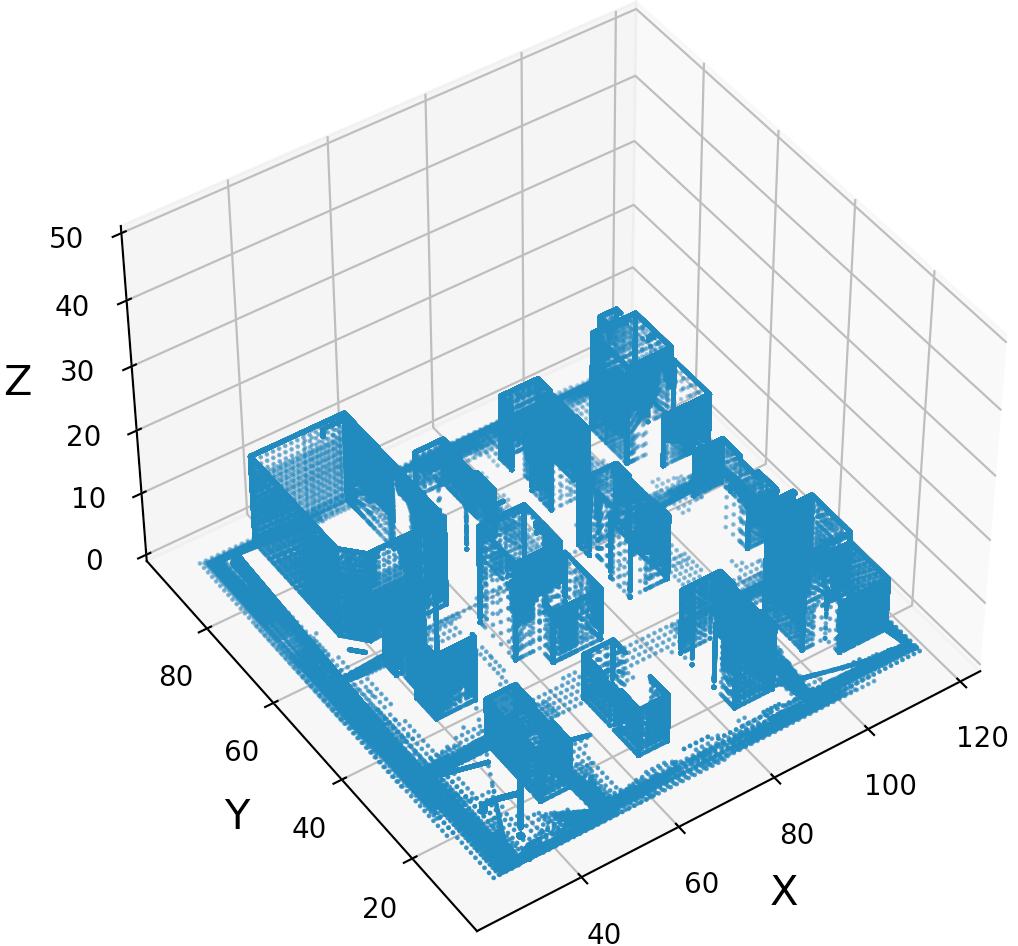}
	\caption{The ground truth of the considered outdoor environment generated directly with the reflections points given by the Wireless Insite.}
	\label{fig:outdoor_ground_truth}
\end{figure}
\begin{figure}
	\centering
	\includegraphics[width=80mm]{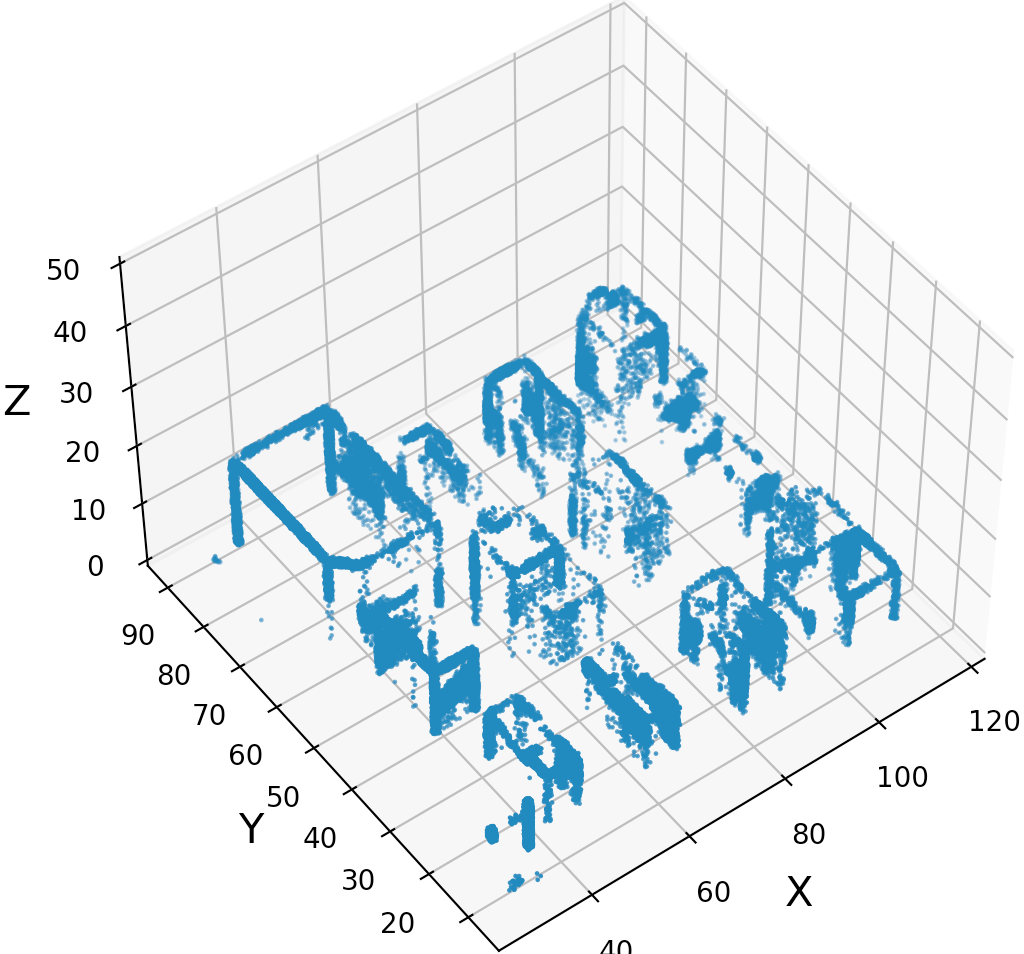}
	\caption{The mapping results of the considered outdoor environment with the C-SLAM.}
	\label{fig:outdoor}
\end{figure}

\begin{figure}
	\centering
	\includegraphics[width=80mm]{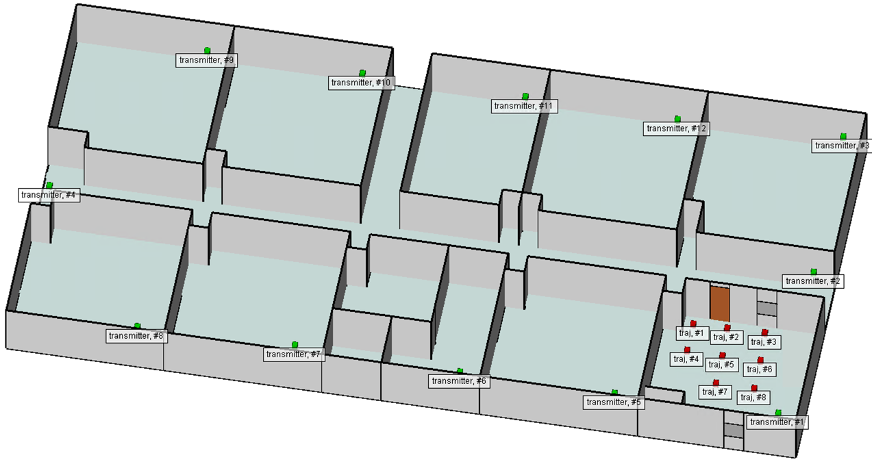}
	\caption{The ground truth of the considered indoor environment.}
	\label{fig:indoor_ground_truth}
\end{figure}

\emph{To answer question (3):} We apply the C-SLAM algorithm to the considered outdoor environment in Fig. \ref{fig:building}, and the map constructed directly using the reflections points given by the Wireless Insite, acting as the ground truth, is shown in Fig. \ref{fig:outdoor_ground_truth}, and the mapping result of the C-SLAM is shown in Fig. \ref{fig:outdoor}. We can see that the C-SLAM is able to build the 3D point cloud map of the considered complex outdoor environment. Notably, the resulted 3D point cloud map contains more details compared with feature-based maps. 
The MSE of the mapping result compared to the ground truth is about $0.3211$m, which is small relative to the whole outdoor environment size. 
We also apply the C-SLAM algorithm to an indoor environment shown in Fig. \ref{fig:indoor_ground_truth}, and the resulted 3D point cloud map is shown in Fig. \ref{fig:indoor}. 
This validates the effectiveness of the C-SLAM in the indoor environment.
\begin{figure}
	\centering
	\includegraphics[width=80mm]{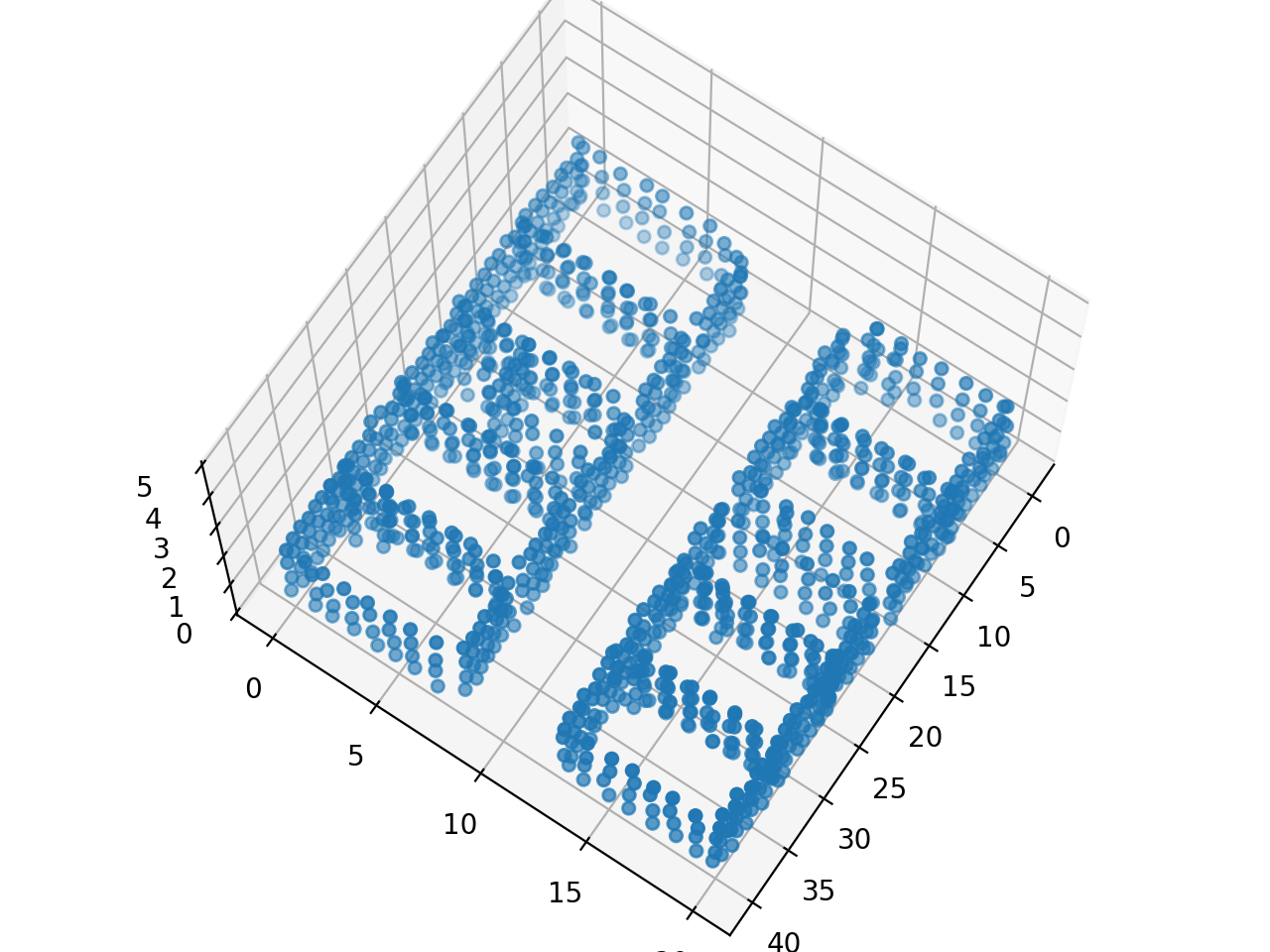}
	\caption{The mapping results of the considered indoor environment with the C-SLAM.}
	\label{fig:indoor}
\end{figure}
\begin{figure}
	\centering
	\includegraphics[width=65mm]{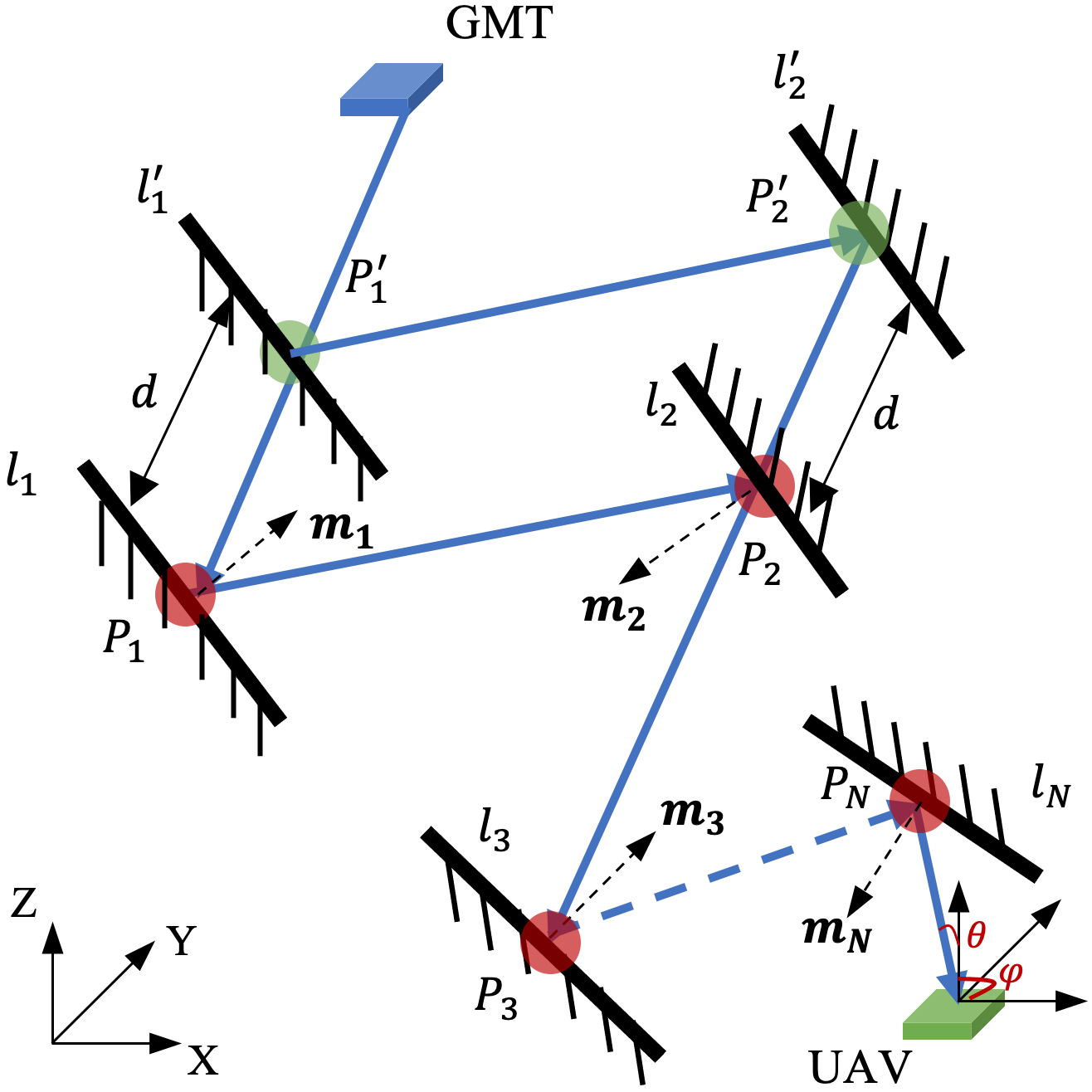}
	\caption{A higher-order NLOS CL from the GMT to the UAV through two or more reflections, where the first two reflection surface $l_1$ and $l_2$ are parallel to each other. Note that when $l_1$ and $l_2$ move the same distance $d$ to $l_1'$ and $l_2'$, the new reflection points $P_1'$ and $P_2'$ also satisfy . As the moving distance $d$ can be an arbitrarily positive number, there are infinite solutions to \eqref{equ:geo_2}.}
	\label{fig:two_reflections}
\end{figure}
\section{Conclusions}
\label{sec:conclusion}
%We study the 3D SLAM problem in complex outdoor environments based only on mmWave wireless communication signals. Firstly, we design the HPC  algorithm that takes the advantages of both the inertial measurement unit and the beam-squint assisted localization method, realizing real-time and precise localizations. Secondly, We propose a DLM algorithm to build the 3D point cloud map of the environment. The basic idea of the DLM is to solve the coordinates of the reflections point on the first-order NLOS CLs and use them as the points in the point cloud map. To achieve this, the DLM designs a classification neural network to identify the first-order NLOS CLs and theoretically calculates the geometric coordinates of the reflection points. Then, combining the HPC and DLM algorithm, we develop the C-SLAM framework that can carry out SLAM without any prior knowledge on the environment. Extensive simulation results validate the effectiveness of our approach.
We study the 3D SLAM problem in complex outdoor and indoor environments based only on mmWave wireless communication signals. 
Firstly, we propose a DLM algorithm 
that can leverage the reflections point on the first-order NLOS CLs to build the 3D point cloud map of the environment. 
Secondly, we take the advantage of both the inertial measurement unit and the beam-squint assisted localization method to realize real-time and precise localizations.
Then, combining the DLM and the adopted localization algorithm, we develop the C-SLAM framework that can carry out SLAM without any prior knowledge on the environment. Moreover, extensive simulations on both complex outdoor and indoor environments validate the effectiveness of our approach.
\appendix
\subsection{Proof of Proposition 1}
\label{app:proof_pro_1}
As shown in Fig. \ref{fig:two_reflections}, consider a higher-order NLOS CL that propagates from the GMT to the UAV through $N\in\mathbb{N}_+$ reflections, where $N\ge 2$. Let the $i$-th reflection point and reflection surface be $P_i$ and $l_i$, respectively, and denote the coordinate of $P_i$ and the normal vector of $l_i$ as $[x_i,y_i,z_i]^T$ and $\mathbf{m}_i$, respectively, $i\in\{1,2,...,N\}$. 
For convenience, we here use $[x_0,y_0,z_0]^T$ and $[x_{N+1},y_{N+1},z_{N+1}]^T$ to represent the coordinates of the GMT and the UAV, respectively, and let points $P_0$ and $P_{N+1}$ be the GMT and the UAV, respectively.
We can then derive the geometry relationship equations as \eqref{equ:geo_2}.
Note that all the coordinates of the reflections points $[x_i,y_i,z_i]$ and the normal vectors $\mathbf{m}_i$ are unknown, and hence, the total number of unknown variables is $6N$. 
In addition, the

\vspace{-4mm}
\begin{strip}
	\hrule
\begin{align}
	\label{equ:geo_2}
	\left\{
	\begin{aligned}
		&\;\;\;\sum_{i=0}^{N} || |\underbrace{[x_{i+1},y_{i+1},z_{i+1}]^T-[x_i,y_i,z_i]^T}_{=\overrightarrow{P_{i}P_{i+1}}}||_2=c\tau,\\
	&<\underbrace{\frac{[x_i,y_i,z_i]^T-[x_{i-1},y_{i-1},z_{i-1}]^T}{\left\|[x_i,y_i,z_i]^T-[x_{i-1},y_{i-1},z_{i-1}]^T\right\|_2}}_{=\overrightarrow{P_{i-1}P_{i}}/\left\|\overrightarrow{P_{i-1}P_{i}}\right\|_2}
		,\mathbf{m}_i> \;=\; <\underbrace{\frac{[x_i,y_i,z_i]^T-[x_{i+1},y_{i+1},z_{i+1}]^T}{\left\|[x_i,y_i,z_i]^T-[x_{i+1},y_{i+1},z_{i+1}]^T\right\|_2}}_{=\overrightarrow{P_{i+1}P_{i}}/\left\|\overrightarrow{P_{i+1}P_{i}}\right\|_2}
			,\mathbf{m}_i>, \;\forall i,\\
		&	<\frac{1}{2}\bigg[\underbrace{\frac{[x_{i-1},y_{i-1},z_{i-1}]^T-[x_i,y_i,z_i]^T}{\left\|[x_{i-1},y_{i-1},z_{i-1}]^T-[x_i,y_i,z_i]^T\right\|_2}}_{=\overrightarrow{P_{i}P_{i-1}}/\left\|\overrightarrow{P_{i}P_{i-1}}\right\|_2}+ 	\underbrace{\frac{[x_{i+1},y_{i+1},z_{i+1}]^T-[x_i,y_i,z_i]^T}{\left\|[x_{i+1},y_{i+1},z_{i+1}]^T-[x_i,y_i,z_i]^T\right\|_2}\bigg]}_{=\overrightarrow{P_{i}P_{i+1}}/\left\|\overrightarrow{P_{i}P_{i+1}}\right\|_2},\mathbf{m}_i>\;=1, \;\forall i\\
		&\;\;\;\left\|\mathbf{m}_i\right\|_2=1, \;\forall i,\\
		&\;\;\; y_N-y_{N+1}=\tan ({\phi})(x_N-x_{N+1}),\\
		&\;\;\; z_{N}-z_{N+1}=\tan\bigg(\frac{\pi}{2}-{\theta}\bigg){\sqrt{(x_N-x_{N+1})^2+(y_N-y_{N+1})^2}},\\
	\end{aligned}
	\right.
\end{align}
	\hrule
\end{strip}
\noindent 
number of equations in \eqref{equ:geo_2} is $3N+3$. 
As $3N+3<6N$ when $N\ge 2$, there are infinite possible solutions to \eqref{equ:geo_2}. Therefore, the coordinates of reflections points on the higher-order NLOS CL cannot be uniquely determined. 
For example, as shown in Fig. \ref{fig:two_reflections} where $l_1$ is parallel to $l_2$, we can translate them along $\overrightarrow{P_{1}P_{0}}$ and $\overrightarrow{P_{2}P_{1}}$ a certain distance $d>0$ at the same time. The reflection points will change from $P_1$ and $P_2$ to $P_1'$ and $P_2'$, respectively. However, $P_1'$ and $P_2'$, together with $P3,...,P_N$ and $\mathbf{m}_1,...,\mathbf{m}_N$, still satisfy \eqref{equ:geo_2}. As $d$ is arbitrary, the number of solutions to \eqref{equ:geo_2} is infinite.

\subsection{Proof of Proposition 2}
\label{app:proof}
From the first equation in \eqref{equ:geo}, we can represent $y_P$ by $x_P$, i.e.,
\begin{align}
	\label{equ:y_P}
	y_P=(x_P-x_R)\tan(\phi)+y_R.
\end{align}
Substitute $y_P$ in the second equation in \eqref{equ:geo} by \eqref{equ:y_P}, we can represent $z_P$ by $x_P$, i.e.,
\begin{align}
	\label{equ:z_P}
	z_P = z_R\;\textcolor{red}{\pm}\;\cot(\theta)(x_P-x_R)\sqrt{1+\tan^2(\phi)}.
\end{align}
Then, substitute $y_P$ and $z_P$ in the third equation in \eqref{equ:geo} by \eqref{equ:y_P} and \eqref{equ:z_P}, we can derive the closed-form expression of $x_P$ as shown in \eqref{equ:derivation}, where $A\triangleq \sqrt{(1+\cot^2(\theta))(1+\tan^2(\phi))}$.
Moreover, we can derive the closed-form expression of $y_P$ and $z_P$ based on the closed-form expression of $x_P$. 
Notably, all $\textcolor{red}{\pm}$'s take the plus or minus sign synchronously, which means that there are two solutions to \eqref{equ:geo}. Specifically, when $x_P\ge x_R$, it takes the plus sign, and when $x_P< x_R$, it takes the minus sign.

\begin{strip}
	\hrule
	\begin{align}
		\label{equ:derivation}
	&	\sqrt{(x_P-x_G)^2+\big(y_R+(x_P-x_R)\tan(\phi)-y_G\big)^2+\big(z_R\;\textcolor{red}{\pm}\;\cot(\theta)(x_P-x_R)\sqrt{1+\tan^2(\phi)}-z_G\big)^2}\notag\\
		&= c\tau  - \bigg(\;\textcolor{red}{\pm}\;(x_P-x_R)\underbrace{\sqrt{1+\tan^2(\phi)+\cot^2(\theta)(1+\tan^2(\phi))}}_{=[(1+\tan^2(\phi))(1+\cot^2(\theta))]^{1/2}\triangleq A}\bigg) 
		\notag\\
		\Rightarrow\qquad & x^2_P-2x_Px_G+x^2_T+(y_R-y_G)^2+2(y_R-y_G)(x_P-x_R)\tan(\phi)+(x_P-x_R)^2\tan^2(\phi)+(z_R-z_G)^2\notag\\
		&\;\textcolor{red}{\pm}\;2\cot(\theta)(x_P-x_R)(z_R-z_G)\sqrt{1+\tan^2(\phi)} + \cot^2(\theta)(x_P-x_R)^2(1+\tan^2(\phi))\notag\\
		&=c^2\tau^2\;\textcolor{red}{\mp}\; 2c\tau A(x_P-x_R) + A^2(x_P-x_R)^2\notag\\
		\Rightarrow\qquad & \frac{1}{2}\underbrace{\bigg(1+\tan^2(\phi)+\cot^2(\theta)(1+\tan^2(\phi))-A^2\bigg)}_{=0 \textit{, (no quadratic form)}}x_P^2+\notag\\
		& \bigg((y_R-y_G)\tan(\phi)-x_G+x_R\underbrace{\big(A^2-\tan^2(\phi)-\cot^2(\theta)(1+\tan^2(\phi))\big)}_{=1}\textcolor{red}{\pm}\cot(\theta)(z_R-z_G)\sqrt{1+\tan^2(\phi)}	\textcolor{red}{\pm}c\tau A\bigg)x_P\notag\\
		&=\frac{1}{2}\bigg[c^2\tau^2 \;\textcolor{red}{\pm}\; 2c\tau Ax_R +\underbrace{A^2x_R^2-x_R^2\tan^2(\phi)-\cot^2(\theta)x_R^2(1+\tan^2(\phi))}_{x^2_R}-x^2_T-(y_R-y_G)^2-(z_R-z_G)^2\notag\\
		&+2(y_R-y_G)x_R\tan(\phi)
		\;\textcolor{red}{\pm}\;2\cot(\theta)x_R(z_R-z_G)\sqrt{1+\tan^2(\phi)}\bigg]
		\notag\\
		\Rightarrow\qquad&
		x_P=\frac{c^2\tau^2 + x_R^2 - x_G^2-(y_R-y_G)^2-(z_R-z_G)^2 + 2\big[(y_R-y_G)\tan(\phi)\;\textcolor{red}{\pm}(z_R-z_G)\cot(\theta)\sqrt{1+\tan^2\phi}\;\textcolor{red}{\pm} c\tau A\big]x_R
		}
	{2\big[\;\textcolor{red}{\pm}\; c\tau A+x_R-x_G\;\textcolor{red}{\pm}\; (z_R-z_G)\cot(\theta)\sqrt{1+\tan^2(\phi)}+(y_R-y_G)\tan(\phi)\big]},
	\end{align}
\hrule
\end{strip}

\end{document}